\documentclass[%
 pra, twocolumn,
 amsmath,amssymb,
]{revtex4-2}

\UseRawInputEncoding

\usepackage{optidef}
\usepackage{braket}
\usepackage{nicematrix}
\usepackage{float}
\usepackage{multirow}
\usepackage{tikz}
\usepackage[utf8]{inputenc}
\usepackage{xcolor}

\usepackage{amsthm,color}

\newtheorem{proposition}{Proposition}

\usepackage{booktabs}
\usepackage{subcaption}
\usepackage{hyperref}
\usepackage{cleveref}

\usepackage{amsmath}
\usepackage{graphicx}% Include figure files
\usepackage{dcolumn}% Align table columns on decimal point
\usepackage{bm}% bold math
\usepackage{ragged2e}

\begin{document}
\crefname{equation}{Eq}{Equations} % capitalize "E", no period
\crefrangelabelformat{equation}{(#3#1#4--#5#2#6)}

\title{Reachability and optimal-time certificates for quantum control}

\author{
  Younes Naceur$^{1,2}$ and Llorenç Balada Gaggioli$^{2,3}$  \\[1em]
  \normalsize{$^{1}$ ICFO -- Institut de Ciències Fotòniques,
08860 Castelldefels (Barcelona), Spain}\\
  \normalsize{$^{2}$ LAAS-CNRS - Laboratoire d’Analyse et d’Architecture des Systèmes, Centre National de la Recherche Scientifique, Toulouse, France}\\
  \normalsize{$^{3}$ Czech Technical University in Prague, Prague, Czech Republic} \\
}

\date{\today}

\begin{abstract}
Finite-time control is central to quantum technologies, yet rigorous limits on reachable targets and optimal control times remain largely unknown. We develop a framework for finite-time reachability and optimal-time certificates in constrained quantum control based on moment relaxations with implicitly time-dependent differential constraints. For fixed control horizons and control constraints, the method yields rigorous upper bounds on achievable terminal fidelities, lower bounds on the optimal control times required to reach them, and certificate gaps for benchmarking explicit control pulses.  We demonstrate the versatility of our framework in three use cases: entangled-state preparation in two and three qubits, one-qubit gate synthesis across different control geometries, and excitation transfer in an $N$-qubit $XX$ chain. Our work establishes differential moment hierarchies as a practical tool for certifying reachability limits and optimal control times in quantum control, providing hardware-aware quantum speed limits while highlighting structure exploitation as a key ingredient for scalable certification.
\end{abstract}

\maketitle

\section*{Introduction}

The development of quantum technologies relies on the ability to manipulate coherent quantum dynamics in a precise, robust, and task-oriented way \cite{Gauger2023LearningQuantumSystems,Awschalom2023QuantumInformationHardware,Preskill2018NISQ}. Across quantum computing, quantum simulation, quantum sensing, and related platforms, one must drive physical systems to prepare states, implement logic operations, suppress unwanted couplings, and transport information under realistic experimental limitations. In this sense, the central challenge is not only to engineer quantum devices, but also to shape their evolution toward prescribed objectives. Quantum control studies how quantum dynamics can be steered to perform concrete operational tasks such as state preparation, gate synthesis, and excitation transfer, and it has become a central tool across spectroscopy, magnetic resonance, quantum information processing, quantum simulation, and quantum sensing \cite{Brif2010QuantumControl,DongPetersen2010Survey,Glaser2015TrainingSchrodingersCat,Koch2022QuantumOptimalControl}. Over the last two decades, this field has developed a rich theoretical and numerical toolbox, ranging from controllability analysis and geometric descriptions of reachable unitary dynamics to highly effective gradient-based pulse-design methods for synthesizing feasible controls \cite{Altafini2002Controllability,Khaneja2005GRAPE,YuanKhaneja2005TimeOptimal,Caneva2011CRAB,Morzhin2019Krotov}. 

These advances, however, also expose a central gap that becomes especially relevant in realistic finite-time settings with hardware constraints. Controllability results often characterize what is possible in principle, while numerical optimal-control routines produce explicit candidate pulses, but neither by itself yields a rigorous statement of the best fidelity actually reachable at a fixed time horizon, nor of the minimum time needed to reach a target fidelity \cite{Koch2022QuantumOptimalControl,Wang2015Brachistochrone}. In particular, when significant control constraints are imposed, the optimization landscape can become very complicated to navigate, so the failure of a local search method cannot be taken as evidence of physical impossibility \cite{Moore2012ConstrainedLandscapes,Pechen2011Traps}. This is precisely where reachability and optimal-time certificates become essential: by supplying rigorous upper bounds on fidelity and lower bounds on the time to reach it, they allow one to distinguish true dynamical limitations from optimization failure, and hence transform finite-time quantum control from a problem of numerical exploration into one of mathematical certification.

Optimization-based methods have approached the modeling of the dynamics in different directions, but none of them provides reachability or optimal-time certificates for quantum control: in classical nonlinear control, occupation measures and moment-SOS hierarchies \cite{lasserre2001global} have been used to obtain global bounds for controlled polynomial dynamics and reachability \cite{LasserreHenrionPrieurTrelat2008Occupation,henrion2024occupation,ROA,gaggioli2026compositiontensortrainstructure}. However, these methods are commutative and therefore cannot encode the noncommutative nature of quantum evolution. In quantum optimal control, global polynomial optimization has recently been combined with polynomial approximations to mitigate the effect of local traps of constrained landscapes for optimal pulse design \cite{g4fb-xm13,BaladaGaggioli2025GateSynthesisPolyOpt,MarecekVala2020MagnusNCPO}. While these methods show that polynomial optimization is a useful tool for quantum control, they optimize over an approximate model, which does not result in rigorous no-go certificates of the original problem. Finally, noncommutative polynomial optimization \cite{npa,navascues_convergent} and its differential extensions \cite{araujo_differential} provide a way of certifying dynamical quantities in noncommutative systems. These results, however, are based on systems under evolution prescribed by a differential equation, and therefore cannot be applied to systems where the time evolution of the control variables is subject to optimization itself.

In this work, we turn this missing link into a practical certification tool for quantum control. For a given hardware model, with amplitude and possibly velocity constraints, and a finite time, our framework certifies what target fidelities are impossible to reach for any admissible pulse. The output is therefore a finite-time reachability envelope: upper bounds rule out fidelities that no admissible pulse can reach by a specific time, while lower bounds from explicit pulses show what actually can be achieved. An associated cumulative relaxation further yields certified lower bounds on the minimum control time needed to reach a target fidelity. With this, we provide a way of bounding the true performance of constrained quantum devices, and obtain rigorous no-go statements about both the fidelities reachable at a given time and the time required to reach a given fidelity. The same comparison also makes the framework an \emph{a posteriori} benchmark for explicit control pulses: given a pulse produced by optimizers such as GRAPE \cite{Khaneja2005GRAPE} or CRAB \cite{Caneva2011CRAB}, reinforcement learning, or experimental calibration, its achieved fidelity can be benchmarked against the certified finite-horizon upper bound under the same hardware constraints.

The method is based on a differential moment formulation of the controlled dynamics, leading to hierarchies of semidefinite relaxations whose values are rigorous bounds on the physical optimum. The computability of the method is based on exact structure exploitation \cite{augier_nicolas_symmetry_2024} in the controlled dynamics, in the spirit of recent advances on sparse, symmetric, and scalable noncommutative polynomial optimization \cite{Klep_2021,wang2021nctssos,wang2024groundstate,wang2026scalable}. In some examples, this reduction turns the problem into a commutative polynomial optimization problem, where these reformulations are exact but simpler descriptions of the same dynamical systems. Across state preparation, gate synthesis, and state transfer, this allows us to certify non-trivial finite-time limitations imposed by the control geometry and the physical system itself, providing rigorous hardware-aware and instance-specific quantum speed limits for constrained quantum technologies. 
%To our knowledge, this is the first moment-hierarchy framework that certifies finite-time unreachable fidelities, as well as lower bounds on optimal quantum control times.

\begin{figure*}[t]
    \centering
    \includegraphics[width=\textwidth,trim=0 5cm 0 5cm,clip]{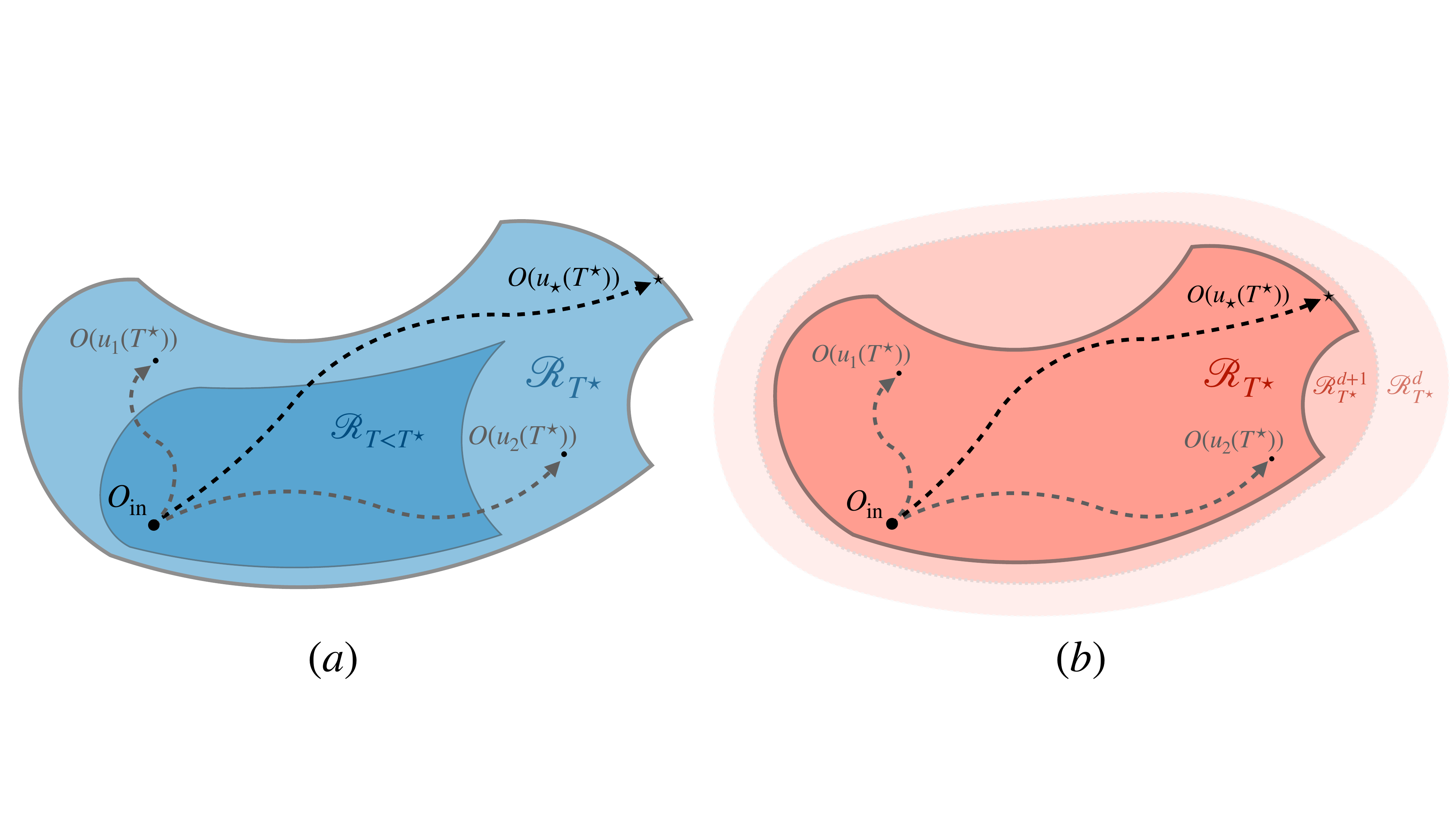}
\caption{
\textbf{Reachability certification and optimal time.}
(a) Geometry of constrained reachability in finite-time quantum control. Starting from the initial observable $O_{\mathrm{in}}$, the admissible controls generate a generally nonconvex reachable set $\mathcal{R}_T$ of terminal observables. The horizon $T^\star$ denotes the optimal threshold time at which the target observable $O_\star$ first becomes reachable, i.e., $O_\star=O(u_\star(T^\star))$ for some optimal admissible control $u_\star(T^\star)$, whereas $O_\star\notin \mathcal{R}_T$ for earlier horizons $T<T^\star$.
(b) The differential moment hierarchy constructs nested convex outer approximations
$\mathcal{R}_{T^\star}^{d+1}\subseteq \mathcal{R}_{T^\star}^{d}$ of the true reachable set $\mathcal{R}_{T^\star}$. If a target observable lies outside one of these outer approximations, it is certified to be unreachable at that horizon.
}
    \label{fig:graphical_abstract}
\end{figure*}

\section{Preliminaries}

\subsection{Quantum control}

We consider a finite-dimensional closed quantum system on a Hilbert space $\mathcal{H}$, driven by a family of controlled Hamiltonians such that
\[
H(t)=H_0+\sum_{j=1}^m u_j(t)H_j,
\]
where $H_0$ denotes the drift Hamiltonian, and the control $u(t)=(u_1(t),\dots,u_m(t))\in W^{1,\infty}([0,T];\mathbb{R}^m)$. To improve the physical relevance of the control modeling, we also consider its velocity
\[
v(t)=(v_1(t),\dots,v_m(t)), \quad \dot{u}_j(t)=v_j(t),\quad j=1,\dots,m,
\]
and impose experimental constraints on both the amplitude and the velocity via the bounds
\[
\sum_{j=1}^mu_j(t)^2\leq U_{\rm{max}}^2, \quad \sum_{j=1}^m v_j(t)^2\leq V_{\rm{max}}^2,
\]
together with the initial condition $u(0)=u_0$.

Given an initial state $\ket{\psi}\in \mathcal{H}$ with $\| \psi \| =1 $, we can describe its dynamics via the Schrödinger equation 
\begin{equation}
    \frac{d}{dt}\ket{\psi(t)}=-iH(t)\ket{\psi(t)}, \quad \ket{\psi(0)}=\ket{\psi},
\end{equation}
or equivalently by the operator equation 
\begin{equation}
    \frac{d}{dt}U(t)=-iH(t)U(t), \quad U(0)=I.
\end{equation}
Therefore, every admissible control pair $(u,v)$ and final time $T$ determine a terminal propagator $U(T)$, and a terminal state $\ket{\psi(T)}=U(T)\ket{\psi}$.

Depending on the control task under consideration, one defines a terminal figure of merit $\Phi(T) = \Phi(O(T))$. 

For the method and most results presented in the paper, the Heisenberg picture is the most natural way of describing the dynamics, for which we take an initial observable $O_{\rm{in}}\in \rm{Herm}(\mathcal{H})$ and define
\[
O(t)=U(t)^{\dag}O_{\rm{in}}U(t),
\]
which satisfies
\begin{equation}
    \frac{d}{dt}O(t)=i[H(t),O(t)], \quad O(0)=O_{\rm{in}}.
\end{equation}

Quantum control can be organized around three main axes. The first one is \textit{reachability} (or more broadly, controllability) \cite{AlbertiniDAlessandro2003,Altafini2003Markovian,AlbertiniDAlessandro2002Spin}, which asks what states, operators or observables can be attained for a fixed time horizon $T$, given initial conditions and a control system. Reachability is especially important in quantum technologies because it connects the mathematical description of a controlled quantum system with the practical capabilities of a physical platform. For a fixed time horizon $T$, we define the constrained reachable objective value $\Phi^{\star}(T):=\sup_{u} \Phi\left(O(u(T))\right)$.

The second axis is that of \textit{optimal control} \cite{Glaser2015TrainingSchrodingersCat,Koch2022QuantumOptimalControl}, the explicit construction of admissible control functions $u_j(t)$ that maximize $\Phi(T)$, driving the system as close as possible to the target objective. Once a task is known to be reachable, one still wishes to determine which admissible controls realize it most accurately and efficiently under the constraints of the platform, but this can be computationally intractable or even undecidable in digitized settings with finitely many available controls \cite{Bondar_2020}. Optimal control therefore provides the constructive counterpart to reachability: it identifies explicit protocols that maximize the chosen performance criterion and translate abstract control objectives into implementable procedures.

The third major theme of quantum control is \textit{robust control} \cite{Weidner2025RobustQuantumControl,Dahleh1990UncertainQuantumSystems,Goldschmidt2022ModelPredictive,Lawrence2026GeometricQuantumControl}, where we consider controls that remain efficient under noise or model mismatch in order to model real quantum devices more accurately, and the optimization problems take the form 
\[
\sup_u \inf_\delta \Phi(O(T,\delta)),
\]
where $\delta$ characterizes the noise or uncertainty, rendering the problem much more challenging computationally. Real quantum devices are affected by noise, calibration errors, and model uncertainty. The goal of robust control is therefore to design controls whose performance remains high even when the implemented dynamics differ from the idealized model.

In this paper we focus on the first topic, reachability, providing rigorous no-go certificates for $\Phi(T)$ under constrained controls, both on the fidelities reachable at a fixed horizon and on the minimum time needed to reach them. In Figure \ref{fig:graphical_abstract} (a), the reachable operator set $\mathcal{R}_T$ at some time $T$ is depicted. From an initial condition $O_{\rm{in}}$, the admissible control functions characterize the reachable set $\mathcal{R}_T$, where the target observable $O_{\star}$ becomes reachable only at a later time $T^{\star}$. This broad subject can be further specialized depending on the task under consideration, where in this work we focus on three classes of tasks.

First, we consider \textit{state preparation} \cite{SklarzTannor2002,Khaneja2005GRAPE}, a fundamental task of quantum protocols, since quantum protocols typically require the system to be initialized in a specific resource state, such as a computational basis state, an entangled state, or the ground state of a target Hamiltonian. For a target projector $P_{\rm{tar}}\in \rm{Herm}(\mathcal{H})$, we define
\[
\mathcal{F}_{\rm{sp}}(T)=\bra{\psi(T)}P_{\rm{tar}}\ket{\psi(T)}.
\]
%since $O(T)=U^{\dag}(T)O(0)U(T)$, if we let $O(0)=P_{\rm{tar}}$.
If $P_{\rm{tar}}=\ket{\psi_{\rm{tar}}}\bra{\psi_{\rm{tar}}}$, then this reduces to the pure-state fidelity
\begin{equation} 
\mathcal{F}_{\rm{sp}}(T)=|\braket{\psi_{\rm{tar}}|\psi(T)}|^2.
\end{equation}

We also consider \textit{gate synthesis} problems \cite{Palao2002UnitaryTransformations}, a central piece of quantum technologies, as quantum information processing is ultimately built from the accurate implementation of prescribed unitary operations. Here, the target is a unitary operator $U_{\star}\in \mathcal{U}(\mathcal{H})$, and the figure of merit is some desired gate-fidelity $\mathcal{F}_{\rm{gate}}(T)$.

The final task we consider is \textit{state transfer} \cite{Jacobs2016FastCommunication}, important because many quantum technologies rely on moving information, excitations, or population between different parts of a device in a controlled and coherent way, which is particularly relevant in architectures where preparation, processing, and readout occur at different locations and must be connected dynamically. These problems are similar to those of state preparation in terms of the terminal objective
\[
\mathcal{F}_{\rm{st}}(T)=\bra{\psi(T)}P_{\rm{tar}}\ket{\psi(T)},
\]
but now $P_{\rm{tar}}$ represents a mode or target site, or any subspace corresponding to a successful transfer.

These three task classes admit a common formulation. In each case, one fixes a control Hamiltonian, an admissible family of controls satisfying the amplitude and velocity constraints, an initial state, and a terminal objective. The target can then be encoded through a terminal performance function $\Phi(T)$ which depends on the task. This unified viewpoint will be the basis for the optimization and certification framework developed in the rest of the paper.

\subsection{Noncommutative polynomial optimization with differential constraints}
\label{sec:ncpo_dc}

We briefly recall the method of moment relaxations for noncommutative polynomial optimization problems (NCPOP) \cite{npa,navascues_convergent,bounding_quantum}, and its extension to differential constraints \cite{araujo_differential}, which builds the foundation for the method developed in this work.

Let $\mathbb{K}\in\{\mathbb{R},\mathbb{C}\}$ and $\mathbb{K}\langle \underline{X}\rangle$ denote the free $*$-algebra generated by a set of noncommuting variables $\underline{X}=(X_1,\dots,X_n)$, equipped with the involution $*$.

A generic instance of an NCPOP reads
\begin{align}
    p^\star=
    \sup_{\{\mathcal{H},\,\ket{\psi},\,\underline{X}\subset\mathcal{B}(\mathcal{H})\}}&
    \quad \bra{\psi} f(\underline{X}) \ket{\psi}
    \label{eq:ncpo_general_moment}\\
    \text{s.t.}\quad
    &g_j(\underline{X})\succeq 0,\quad j=1,\dots,m_G,\nonumber\\
    &h_i(\underline{X})=0,\quad i=1,\dots,m_H,\nonumber\\
    &\bra{\psi}q_k(\underline{X})\ket{\psi}=0,
    \quad k=1,\dots,m_Q.\nonumber
\end{align}
Here the optimization ranges over all Hilbert spaces $\mathcal{H}$, normalized vectors $\ket{\psi}\in\mathcal{H}$, and feasible tuples of bounded self-adjoint operators. This problem is highly non-convex and, in general, intractable \cite{ji2022mipre}. Following \cite{npa,navascues_convergent,bounding_quantum}, one lifts it to an optimization over moment functionals
\begin{equation}
L_\psi(f)=\bra{\psi}f(\underline X)\ket{\psi},
\end{equation}
or equivalently over moment sequences $\psi_w=L_\psi(w)$ indexed by words $w$ in the noncommuting variables.

Now, to make this problem tractable, one replaces the exact set with a larger, convex set of \emph{pseudo-moment} sequences $(y_w)$. These are collections of numbers 
\begin{equation}
    y_w =L_y(w),\qquad L_y: \mathbb{K}\langle \underline{X}\rangle\to\mathbb{K},
\end{equation}

which satisfy a set of necessary conditions that every genuine quantum moment sequence must obey. Imposing iteratively more necessary conditions on $L_y$ defines a hierarchy of approximations of the set of proper quantum moments, which, under mild conditions, converges to the exact set as $d\rightarrow \infty$ \cite{npa}.

At relaxation order $d$, state positivity is enforced on the truncated algebra $\mathbb K\langle \underline X\rangle_{2d}$ through the moment matrix
\begin{equation}
M_d(y)(a,b)=L_y(a^\ast b),\qquad a,b\in\mathbb K\langle\underline X\rangle_d,
\end{equation}
together with localizing matrices for polynomial positivity constraints
\[
M_{d-d_j}(g_j\,y)(a,b)=L_y(a^*g_j b),
\qquad
d_j=\left\lceil \frac{\deg(g_j)}{2}\right\rceil .
\]
Equality constraints are usually imposed by working in the quotient algebra $\mathbb{K}\langle \underline{X}\rangle_{2d} /\mathcal{I}$ over the ideal $\mathcal{I}$ generated by the constraints $h_i(\underline{X})=0$, e.g., indexing $M_d(y)$ with reduced words $a,b\in \mathbb{K}\langle \underline{X}\rangle_{2d} /\mathcal{I}$. Any additional moment-level constraints are imposed as linear relations on the truncated sequence $(y_w)_{|w|\leq 2d}$. One thus arrives at the order-$d$ moment relaxation of problem \eqref{eq:ncpo_general_moment}:
\begin{align}
    p_d=\sup_y\quad &L_y(f)
    \label{eq:ncpo_moment_general}\\
    \text{s.t.}\quad
    &M_d(y)\succeq 0,\nonumber\\
    &M_{d-d_j}(g_j\,y)\succeq 0,\qquad j=1,\dots,m_G,\nonumber\\
    &L_y(a^*h_ib)=0,\qquad i=1,\dots,m_H,\nonumber\\
    &L_y(q_k)=0,\qquad k=1,\dots,m_Q,\nonumber\\
    &L_y(1)=1.\nonumber
\end{align}
Every feasible quantum realization of \eqref{eq:ncpo_general_moment} yields a feasible pseudo-moment sequence for \eqref{eq:ncpo_moment_general}, so $p_d$ is an upper bound on the true optimum $p^\star$. The advantage is that \eqref{eq:ncpo_moment_general} is a semidefinite program (SDP), making it computationally tractable through a wide range of solvers for such problems \cite{mosek}.

Now, in \cite{araujo_differential} it is shown that one can further generalize the class of problems \eqref{eq:ncpo_general_moment} by allowing \textit{differential constraints} of the form
\begin{equation}
\label{eq:diffeq}
    \frac{dX_i(t)}{dt}=\zeta_i(\underline{X}) \, \text{ for } \, i \in \mathcal{I}_{\mathcal{D}},
\end{equation}
where $\mathcal{I}_{\mathcal{D}}$ denotes the set of variables following some differential equation. This allows one to obtain bounds on system observables following some dynamical evolution, or to model non-polynomial functions emerging as solutions of differential equations. To deal with such constraints, one introduces a formal commuting variable $\tau \in [0,1]$ and, instead of a single static moment functional $L_y$, considers a formally time-dependent linear functional $L_{y(\tau)}$ together with boundary functionals $L_{y(0)}$ and $L_{y(1)}$. The algebra $\mathbb{K}\langle\underline{X},\tau\rangle$ is then generated by the time-dependent variables $X_i(\tau)$ following some differential equation, the time-independent variables $\{X_i: i\notin \mathcal{I}_{\mathcal{D}} \}$, and the commuting variable $\tau$ itself.

% mention that one can model commuting ODE by enforcing commutation with all other variables.

Now, to impose differential identities \eqref{eq:diffeq} on the level of the moment matrices, one first defines a differential map $\mathcal{D}:\mathbb{K}\langle \underline{X},\tau\rangle \rightarrow \mathbb{K}\langle \underline{X},\tau\rangle$ by
\begin{align}
    &\mathcal{D}(X_i) = \zeta_i(\underline{X}),\nonumber \qquad i\in\mathcal{I}_{\mathcal{D}},\\
    &\mathcal{D}(X_i) = 0, \qquad i\notin\mathcal{I}_{\mathcal{D}},\nonumber\\
    &\mathcal{D}(\tau)=1, \nonumber \\
    &\mathcal{D}(hh') = \mathcal{D}(h)h' + h\mathcal{D}(h'),\qquad h,h' \in \mathbb{K}\langle \underline{X},\tau\rangle. \nonumber
\end{align}
By making use of the initial differential equations \eqref{eq:diffeq} and the Leibniz rule, this defines a differential map on the full truncated algebra $\mathbb{K}\langle \underline{X},\tau\rangle_{2d}$. Differential constraints are then incorporated by imposing the moment equivalent of the fundamental theorem of calculus:
\begin{equation}
    L_{y(\tau)}\left(\mathcal{D}(p)\right)-L_{y(1)}(p)+L_{y(0)}(p) = 0.
    \label{eq:general_ibp}
\end{equation}
At relaxation order $d$, a generic dynamical noncommutative polynomial optimization problem (DNPO) is then relaxed to the SDP
\begin{align}
&\sup_{y(\tau),\,y(0),\,y(1)} \quad  L_{y(1)}(f) \label{eq:dnpo_relaxation}\\
\text{s.t.}\quad
 &M_d\!\left(y(\tau)\right) \succeq 0,\quad
  M_d\!\left(y(0)\right) \succeq 0,\quad
  M_d\!\left(y(1)\right) \succeq 0, \nonumber\\
& M_{d-d_j}\!\left(g_j\,y(\tau)\right) \succeq 0,\qquad j=1,\dots,m_G, \nonumber\\
& L_{y(\tau)}(a^* h_i b) = 0,\qquad i=1,\dots,m_H, \nonumber\\
& L_{y(\tau)}(q_k) = 0,\qquad k=1,\dots,m_Q, \nonumber\\
& L_{y(\tau)}\!\left(\mathcal{D} (p)\right) = L_{y(1)}(p) - L_{y(0)}(p),
\ \forall p \in \mathbb{K}\langle \underline{X},\tau\rangle_{2d}, \nonumber\\
& L_{y(0)}(1)=L_{y(1)}(1)=L_{y(\tau)}(1)=1. \nonumber
\end{align}
Here $M_d(y(\tau))$, $M_d(y(0))$, and $M_d(y(1))$ denote the moment matrices associated, respectively, with the interior and boundary pseudo-moment functionals, and $M_{d-d_j}(g_j y(\tau))$ the corresponding localizing matrices. As in the static case, equality constraints may equivalently be imposed by working in the appropriate quotient algebra. The same differential framework also covers the commutative case \cite{LasserreHenrionPrieurTrelat2008Occupation}, obtained by imposing the appropriate commutation relations. 

The framework in this paper takes \cite{araujo_differential} as a starting point but generalizes it to controlled DNPOs: in quantum control, the central object is not a single predefined time evolution, but a whole reachable set generated by all admissible control functions under device constraints. In the next section, we therefore build controlled noncommutative differential moment hierarchies in which the quantum dynamics, control functions, and velocities are treated simultaneously by imposing the equations of motion as differential constraints, and the amplitude and velocity bounds as constraints on the admissible pulses. This turns \eqref{eq:dnpo_relaxation} into a finite-horizon reachability and optimal-time certification framework, whose SDP relaxations bound the best possible fidelity attainable and the lowest time for which a fidelity is reachable by any admissible control, and hence certify finite-time impossibility statements for constrained quantum control.

\section{Method}
We now describe the certification pipeline used throughout the paper to obtain rigorous finite-horizon bounds on reachable fidelities and optimal control times, where we first introduce a hierarchy for reachable fidelities, and then a cumulative version for optimal control times.
\subsection{Reachability certificates}
% should we include plots where we alter U_max or V_max and see if reachability differs?
 We consider reachability problems of the following general form
% do we have to sometimes work in Schrödinger, sometimes in heisenberg picture?
\begin{align}
    &\sup_{\{\ket{\psi},\,u_j,\,v_j,\,O(t)\}}
    \quad
    \Phi(O(T))
    \label{eq:general_reachability_heisenberg}\\
    \text{s.t.}\quad
    &\frac{d}{dt} O(t)
    =
    i\bigl[H(t),\,O(t)\bigr],
    \qquad t\in[0,T],
    \nonumber\\
    &H(t)
    =
    H_0+\sum_{j=1}^m u_j(t)\,H_j,
    \nonumber\\
    &\frac{du_j(t)}{dt} = v_j(t),
    \qquad j=1,\dots,m,
    \nonumber\\
    &\sum u_i(t)^2 \leq U_{\max}^2, \qquad \sum v_i(t)^2 \leq V_{\max}^2 \nonumber\\
    &u(0)=u_0, \qquad O(0)=O_{\mathrm{in}},\nonumber 
\end{align}
where we optimize some terminal figure of merit $\Phi(O(T))$ over possible pure quantum states $\ket{\psi}$, time-dependent control functions $u(t)$, their derivatives $v(t)$, and physical operator trajectories $O(t)$. The way we model $\Phi(O(T))$ as a polynomial in the problem variables depends on the specific problem instance. Likewise, although the dynamics in \eqref{eq:general_reachability_heisenberg} are written in Heisenberg form for generality, the concrete benchmarks below may be expressed in whichever exact representation is most convenient, for instance directly in the Schr\"odinger picture. This general problem class naturally accommodates further constraints  
as in \eqref{eq:ncpo_general_moment}.

The reachability problem \eqref{eq:general_reachability_heisenberg} is close in spirit to the DNPO framework of \cite{araujo_differential}, but differs in an important structural point. In the standard DNPO setting, each time-dependent generator is assigned a prescribed polynomial vector field, and time-independent ones are assigned the trivial vector field. In constrained quantum control, by contrast, the physical variables and the control amplitudes have prescribed dynamics,
\[
\frac{dO(t)}{dt}=i[H(t),O(t)],\qquad \frac{du_j(t)}{dt}=v_j(t),
\]
but the velocity variables $v_j(t)$ are free time-dependent decision functions. They are constrained algebraically, for example by $\sum_j v_j(t)^2\le V_{\max}^2$, but no acceleration model $d v_j/dt=a_j$ is imposed. Thus the relaxation introduced below is a mixed differential/occupation-moment relaxation: differential constraints are imposed only on variables with prescribed evolution, while the velocities enter as free time-dependent generators subject to pointwise polynomial constraints, e.g. indexing only the interior moment matrix corresponding to $L_{y(\tau)}$. We hence restrict the set of test polynomials on which the recursive differential constraints are imposed to the ones composed of variables with prescribed (possibly trivial) dynamics:
\begin{equation}
    \mathcal{T}^d := \{p\in \mathbb{K}\langle O_i(\tau),u_i(\tau),\tau \rangle \, | \, \mathrm{deg}(p)\leq 2d \}.
\end{equation}

We adapt the theory from \cite{araujo_differential} to deal with such  implicitly time-dependent differential constraints. Note that this framework also allows for unconstrained velocities, where the variables $v_j$ and associated constraints can be eliminated and the controls $u_j$ take the role of free time-dependent functions.

The differential framework is naturally expressed in normalized times $\tau\in[0,1]$. To pass to physical times $t\in[0,T]$ it is convenient to introduce the terminal horizon $T$ as an additional problem parameter:
\begin{equation}
    t = T \tau, \qquad \tau \in [0,1],
\end{equation}
which lets us rewrite the differential constraints \eqref{eq:general_reachability_heisenberg},
\begin{equation}\small
    \frac{d}{d\tau} O(\tau)
    =
    iT\bigl[H(\tau),\,O(\tau)\bigr],
    \quad \frac{du_j(\tau)}{d\tau }=Tv_j(\tau),\quad\tau\in[0,1].
\end{equation}
The set of generators $\mathcal{G}_\tau$ of $\mathbb{K}\langle \underline{X},\tau \rangle $ then takes the following form
\begin{equation}
    \mathcal{G}_\tau = \{\tau, u_i(\tau), v_i(\tau), O_i(\tau) \}.
\end{equation}

A generic word $w \in \mathbb{K}\langle \underline{X},\tau \rangle$ generated by $\mathcal{G}_\tau$ can then be written as
\begin{equation}
    w = \tau^a\left(\prod_iu_i(\tau)^{b_i}\right)\left(\prod_j v_j(\tau)^{c_j} \right)w_{O(\tau)},
\end{equation}
where $a,b_i,c_j\in \mathbb{N}$ characterize the commuting monomials, and $w_{O(\tau)}\in \mathbb{K}\langle \underline{X}\rangle $ are noncommutative monomials in the problem variables $O_i(\tau)$. In each concrete benchmark, this generic basis is then specialized further by exploiting exact reductions such as symmetry restriction, or representation changes. In some examples, these reductions eliminate the noncommutative nature of the operator structure entirely, yielding a commutative polynomial optimization problem with differential constraints, which is a special case of the same moment-relaxation framework.

To construct our control relaxations, we roughly follow the construction from \cite{araujo_differential}:
control bounds are encoded as positivity constraints of the form $g_j(\underline{X})\succeq 0$; algebraic identities of the model are imposed as quotient relations $h_i(\underline{X})=0$; and the equations of motion are incorporated through the differential map $\mathcal{D}$ together with the moment form of the fundamental theorem of calculus \eqref{eq:general_ibp} on polynomials $p\in \mathcal{T}^d$. We emphasize again that these differential constraints are imposed only on variables following some evolution, e.g. particularly not on $v_i(t)$.

At relaxation order $d$, this problem is a direct extension of \eqref{eq:dnpo_relaxation} with additional implicitly time-dependent evolution equations for $u_i(t)$, and the SDP relaxation can be written as

\begin{align}\label{eq:SDP_reach}
\Phi^d(T&)=
\sup_{y(\tau),y(0),y(1)}\quad
 L_{y(1)}(f) \\[1mm]
\mathrm{s.t.}\quad
& M_d(y(\tau))\succeq 0,\quad
  M_d(y(0))\succeq 0,\quad
  M_d(y(1))\succeq 0,
\nonumber\\
& M_{d-1}( \alpha(1-\alpha)\,y(\cdot))\succeq 0, \, \,\mathrm{for}\,(\cdot)\in \{0,\tau,1 \}
\nonumber\\
& M_{d- d_j}( g_j\,y(\cdot))\succeq 0, \,\mathrm{for}\,(\cdot)\in \{0,\tau,1 \},
  \quad g_j\in\mathcal{T}^d, 
\nonumber\\
& M_{d- d_j}( g_j\,y(\tau))\succeq 0, \,
  \quad g_j\notin\mathcal{T}^d, 
\nonumber\\
& L_{y(\cdot)}(a^*h_i b)=0,
  \ i=1,\ldots,m_H, \text{ for } (\cdot) \in \{0,\tau,1 \}\nonumber\\
& L_{y(\tau)}(\mathcal{D}(p))=L_{y(1)}(p|_{\tau=1})-L_{y(0)}(p|_{\tau=0}),
  \ p\in\mathcal{T}^d,\nonumber\\
& L_{y(\tau)}(1)=L_{y(0)}(1)=L_{y(1)}(1)=1,\nonumber\\
& L_{y(0)}(w)=w_{\mathrm{in}},\nonumber
\end{align}
where $w_{\mathrm{in}}$ are fixed by the initial data and $f$ is the polynomial expression for $\Phi(O(T))$. The interior moment matrix is indexed by monomials in $\mathbb{K}\langle \underline{X},\tau\rangle$ while the boundary moment matrices are indexed by monomials in $\mathcal{T}^d$. Localizing matrices are built for all three moment functionals $L_{y(\cdot)}$ for inequality constraints with prescribed evolution $g_j\in \mathcal{T}_d$, and only on the interior functional $L_{y(\tau)}$ for the $g_j$ featuring free time-dependent variables such as $v_j$. The following proposition states the soundness and monotonicity of this moment hierarchy.

\begin{proposition}[Finite-horizon reachability certificates from implicitly time-dependent DNPO]
\label{prop:reachability_certificate}
Fix a control horizon $T$ and consider a finite-$T$ reachability problem of the form \eqref{eq:general_reachability_heisenberg} after any \emph{exact} reformulation or reduction of the problem into a polynomial optimization problem with implicitly time-dependent differential constraints. Let $\Phi^{\star}(T)$ denote the corresponding optimal reachable value, and let $\Phi^d(T)$ denote the optimal value of the corresponding order-$d$ relaxation.

Then every physically feasible controlled trajectory induces a feasible moment tuple for \eqref{eq:SDP_reach}. Consequently,
\begin{equation}
    \Phi^{\star}(T)\leq \Phi^d(T)
\end{equation}
for every relaxation order $d$. Furthermore, the upper bounds are monotonically improving
\begin{equation}
    \cdots \leq \Phi^{d+1}(T)\leq\Phi^d(T) \leq \cdots 
\end{equation}

\end{proposition}

\begin{proof}
    Fix a time horizon $T$, and let $u(t)$ and the corresponding reduced variables $O(t)$ define a physically feasible controlled trajectory, $\gamma(\tau)$, of the polynomial control problem obtained from \eqref{eq:general_reachability_heisenberg}. Let $f$ be the polynomial expression in reduced variables of the objective function, such that evaluating $f$ at the terminal point of the trajectory reproduces the expected physical objective value.

    We can define the associated truncated moment functionals $L_{y(\tau)},L_{y(0)},L_{y(1)}$ by evaluating polynomials on the given trajectory and, in the noncommutative case, taking the expectation value. We define polynomials on the whole trajectory to be $p:=p(\tau,u,v,O)$, and polynomials on the boundaries as $p_{\mathrm{b}}:=p(u,O)\in \mathcal{T}^d$, acting only on variables with well-defined endpoint values. To be precise
    \begin{align*}
        L_{y(\tau)}(p)&=\int_0^1 p(\gamma(\tau))d\tau, \\ L_{y(0)}(p_{\mathrm{b}}) = p_{\mathrm{b}}|_{\tau=0},&\quad L_{y(1)}(p_{\mathrm{b}}) = p_{\mathrm{b}}|_{\tau=1}.
    \end{align*}

    We now show that these moments are feasible for \eqref{eq:SDP_reach}. First, note that normalization holds by construction, as $L_{y(\tau)}(1)=L_{y(0)}(1)=L_{y(1)}(1)=1$, and boundary constraints hold because $L_{y(0)}$, $L_{y(1)}$ are the moment functionals of Dirac masses at the initial and terminal points of the admissible trajectory.

    Every bulk positivity constraint is satisfied almost everywhere along the trajectory. Therefore, for every admissible test polynomial $r$ we have $L_{y(\tau)}(r^*g_j r)\geq 0$. Boundary positivity constraints are imposed only for constraints involving variables with well-defined endpoint values, which hold because the boundary functionals are Dirac evaluations at the endpoints. Therefore, the corresponding moment and localizing matrices are positive semidefinite.

    The equality constraints $h_i(\underline{X})=0$ also hold pointwise along the whole trajectory, so it follows that $L_{y(\tau)}(a^*h_i b)=0$ for all feasible words $a,b$, and similarly for the boundary points. Therefore, the equality constraints of \eqref{eq:SDP_reach} are satisfied.

    Finally, the differential constraints follow from the chain rule and the definition of the map $\mathcal{D}$, where differential identities are only imposed on admissible test polynomials \(p=p(\tau,u,O)\) that do not contain the free velocity variables \(v_j\), although \(\mathcal D(p)\) may contain them. Along the feasible trajectory, and for any admissible polynomial $p$ on variables with prescribed differential evolution,
    \begin{equation*}
        \frac{d}{d\tau}p=\mathcal{D}(p).
    \end{equation*}
    Integrating and applying the fundamental theorem of calculus, we get 
    \begin{equation*}
        L_{y(\tau)}(\mathcal{D}(p))=L_{y(1)}(p_{\mathrm{b}})-L_{y(0)}(p_{\mathrm{b}}),
    \end{equation*}
    which is exactly the differential constraint imposed in \eqref{eq:SDP_reach}.

    Thus every physically feasible controlled trajectory induces a feasible moment tuple for the order-$d$ DNPO relaxation. Moreover, by construction, the terminal objective is preserved $L_{y(1)}(f)=\Phi(O(T))$.

    Therefore every admissible physical trajectory gives a feasible point of the relaxation with the same objective value. Since the relaxation optimizes over all such feasible moment tuples, its optimum is an upper bound on the true reachable value. Hence,
    \begin{equation*}
        \Phi^{\star}(T)\leq \Phi^d(T).
    \end{equation*}

    To prove monotonicity, note that the order-$(d+1)$ relaxation contains all moments up to degree $2(d+1)$ and imposes all constraints of the order-$d$ relaxation, together with other moment, localizing and differential constraints from other higher-degree test polynomials. Therefore, if a moment tuple is feasible at order $d+1$, its restriction to moments up to degree $2d$ is feasible for order-$d$. Therefore, the feasible set of the order-$(d+1)$ projects onto the feasible set of order-$d$ relaxation, and since both problems are maximization problems of the same objective, it follows that
    \[
    \dots\leq \Phi^{d+1}(T)\leq \Phi^d(T)\leq \dots.
    \]
    This proves that the upper bounds are monotonically improving as the relaxation order increases.
    
\end{proof}

Geometrically, the hierarchy yields a nested family of convex outer approximations to the true fixed-horizon reachable set, as illustrated schematically in Fig.~\ref{fig:graphical_abstract}(b).

\subsection{Optimal time lower bounds}
One can slightly modify the hierarchy presented in Proposition \ref{prop:reachability_certificate}, to obtain lower bounds on optimal control times required to reach a certain target. The main technical adjustment lies in changing from certifying reachability at a \emph{fixed} horizon $T$ to a cumulative problem, certifying all physical times $0\leq s \leq T$.

Starting with the problem formulation, one asks for the optimal control time to reach a terminal fidelity larger than some target threshold $\eta$,
\begin{align}
    &\inf_{\{\ket{\psi},\,u_j,\,v_j,\,O(t)\}}
    \quad
    T
    \label{eq:general_optimal_time}\\
    \text{s.t.}\quad
    &  \Phi(O(T))\geq \eta, \nonumber\\
    &\frac{d}{dt} O(t)
    =
    i\bigl[H(t),\,O(t)\bigr],
    \qquad t\in[0,T],
    \nonumber\\
    &H(t)
    =
    H_0+\sum_{j=1}^m u_j(t)\,H_j,
    \nonumber\\
    &\frac{du_j(t)}{dt} = v_j(t),
    \qquad j=1,\dots,m,
    \nonumber\\
    &\sum u_i(t)^2 \leq U_{\max}^2, \qquad \sum v_i(t)^2 \leq V_{\max}^2 \nonumber\\
    &u(0)=u_0, \qquad O(0)=O_{\mathrm{in}}.\nonumber 
\end{align}

To relax the optimal-time problem directly, we introduce a parameter $\overline{T}$ as the maximal search horizon for optimal stopping times,  $s\in[0,\overline{T}]$, and introduce a rescaled stopping time variable $\alpha\in[0,1]$, so that the physical terminal time is $s=\alpha \overline{T} $. The optimal-time problem is then relaxed over trajectories in normalized
time $\tau\in[0,1]$, with dynamics scaled by the unknown stopping time
$\alpha \overline{T}$:
\begin{align}
    \frac{dO(\tau)}{d\tau}& = i\alpha \overline{T}[H(\tau),O(\tau)] \label{eq:diffeqalpha}\\
    \frac{du_j(\tau)}{d\tau}& = \alpha \overline{T} v_j(
\tau) \nonumber
\end{align}

We set $\mathcal{D}(\alpha) = 0$, fixing its explicit time-independence. As before, the differential constraints are imposed only on admissible test polynomials involving variables with prescribed evolution, now of the form $p=p(\tau,\alpha,u,O)$. 

At relaxation order $d$, since the physical stopping time is $s=\alpha\overline{T}$, the relaxed objective is $\overline T\,L_{y(1)}(\alpha)$, and the optimal-time relaxation is

\begin{align}\label{eq:sdp_time}
T^d_\eta(\overline{T})&
=
\inf_{y(\tau),y(0),y(1)}\quad
 \overline T\,L_{y(1)}(\alpha)
\\[1mm]
\mathrm{s.t.}\quad
& L_{y(1)}(f)\ge \eta,
\nonumber\\
& M_d(y(\tau))\succeq 0,\quad
  M_d(y(0))\succeq 0,\quad
  M_d(y(1))\succeq 0,
\nonumber\\
& M_{d-1}( \alpha(1-\alpha)\,y(\cdot))\succeq 0, \, \,\mathrm{for}\,(\cdot)\in \{0,\tau,1 \}
\nonumber\\
& M_{d- d_j}( g_j\,y(\cdot))\succeq 0, \,\mathrm{for}\,(\cdot)\in \{0,\tau,1 \},
  \quad g_j\in\mathcal{T}^d, 
\nonumber\\
& M_{d- d_j}( g_j\,y(\tau))\succeq 0, \,
  \quad g_j\notin\mathcal{T}^d, 
\nonumber\\
& L_{y(\cdot)}(a^*h_i b)=0,
  \quad i=1,..,m_H,\text{ for } (\cdot) \in \{0,\tau,1 \}
\nonumber\\
& L_{y(\tau)}\left(\mathcal{D}\left(p\right)\right)
  =
  L_{y(1)}(p|_{\tau=1})
  -
  L_{y(0)}(p|_{\tau=0}),
  \ p\in\mathcal T^d,
\nonumber\\
& L_{y(\tau)}(1)=L_{y(0)}(1)=L_{y(1)}(1)=1,
\nonumber\\
& L_{y(0)}(w)=w_{\mathrm{in}}.\nonumber
\end{align}

Here $f$ is the polynomial representative of the terminal objective $\Phi$, so the constraint $L_{y(1)}(f)\ge\eta$ is the moment form of the target condition $\Phi(O(s))\ge\eta$. Note that $\overline{T}$ should be chosen larger than the expected optimal control time, since if $T^\star_\eta > \overline{T}$, the corresponding problem is infeasible and will certify only $ \overline{T}<T^\star_\eta$. A convenient choice of $\overline{T}>T^\star_\eta $ is any available feasible control temporal data obtained through explicit control pulses reaching threshold $\eta$. Empirically, we observe that $\overline{T}$ should not be chosen too large, as the quality of the bounds depends on the \textit{tightness} of the chosen search window, see Section \ref{sec:statetransfer} for details. 

The constraint $\alpha\in [0,1]$ is implemented via the corresponding localizing matrices. The set $\mathcal T^d$ again denotes the admissible differential test polynomials in the variables $(\tau,\alpha,u,O)$, chosen so that both $p$ and $\mathcal{D}(p)$ are represented at relaxation order $d$. The last constraints impose the initial data $u(0)=u_0$ and $O(0)=O_{\mathrm{in}}$ in moment form, now in the
enlarged algebra including the static variable $\alpha$.

\begin{proposition}[Certified lower bounds on optimal control times]
\label{prop:threshold_time}
Fix a target objective threshold \(\eta\in\mathbb{R}\) and define the
true optimal control time
\[
T^\star_\eta
:=
\inf\left\{
T\ge 0:\  \Phi(O(T))\ge \eta
\right\}.
\]
Fix a maximal search horizon $\overline T>0$, and let
$T^d_\eta(\overline T)$ denote the value of the order-$d$
optimal-time relaxation if feasible. Then every admissible controlled trajectory reaching the threshold
$\eta$ up to some time $\overline T$ induces a feasible moment tuple
for the order-$d$ optimal-time relaxation.
Consequently,
\[
T^d_\eta(\overline T)\le T^\star_\eta,
\]
whenever $T^\star_\eta\le \overline T$. Hence
$T^d_\eta(\overline T)$ is a certified lower bound on the true optimal
control time.

If the relaxation is infeasible, then no admissible physical trajectory
can reach the threshold $\eta$ by time $ \overline T$, and
therefore
\[
\overline T\le T^\star_\eta.
\]
\end{proposition}

\begin{proof}
Fix $\overline{T}> 0$ and let $s\in[0,\overline{T}]$. Consider any admissible physical trajectory on $[0,s]$, with controls $u(t)$, velocities $v(t)$, and variables $O(t)$. Set $\alpha=s/\overline{T}\in[0,1]$ and reparametrize by $t=s\tau=\alpha \overline{T}\tau$, $\tau\in[0,1]$.

Under this reparametrization, one obtains the rescaled equations of motion \eqref{eq:diffeqalpha}. The amplitude and velocity constraints are unchanged, since they are imposed in physical time and hold pointwise along the original trajectory, and $\alpha\in[0,1]$ satisfies the localizing constraint $\alpha(1-\alpha)\ge0$. Defining the moment functionals by evaluation along this trajectory, exactly as in the proof of Proposition~\ref{prop:reachability_certificate}, all moment, localizing, equality and differential constraints of the order-$d$ relaxation are satisfied. The terminal objective is preserved, $L_{y(1)}(f)=\Phi(O(s))$, and the relaxed time objective evaluates to
\[
\overline{T}\,L_{y(1)}(\alpha)=\overline{T}\,\alpha=s,
\]
the physical stopping time of the trajectory.

\emph{Lower bound.} Suppose $T^\star_\eta\le\overline{T}$. By definition of $T^\star_\eta$ there exists an admissible trajectory reaching $\Phi(O(s))\ge\eta$ at $s=T^\star_\eta$. By the construction above, this trajectory induces a feasible moment tuple for the order-$d$ relaxation, since it satisfies the target constraint $L_{y(1)}(f)=\Phi(O(s))\ge\eta$, with objective value $\overline{T}\,L_{y(1)}(\alpha)=T^\star_\eta$. As $T^d_\eta(\overline{T})$ minimizes the same objective over the larger feasible set of all moment tuples,
\[
T^d_\eta(\overline{T})\le T^\star_\eta.
\]

\emph{Infeasibility.} Conversely, suppose the order-$d$ relaxation is infeasible. If some admissible trajectory reached $\Phi(O(s))\ge\eta$ at a stopping time $s\le\overline{T}$, the construction above would yield a feasible moment tuple, contradicting infeasibility. Hence no admissible control reaches the threshold $\eta$ at any time $s\le\overline{T}$, and therefore, by the definition of $T^\star_\eta$,
\[
\overline{T}\le T^\star_\eta. \qedhere
\]

\end{proof}

Note that if we let $\overline{T}\rightarrow \infty$, the infeasibility of the optimization problem \eqref{eq:general_optimal_time} for all  $s\in[0,\overline{T}]$ also certifies   uncontrollability of the system of interest, which constitutes a result of independent interest.  However, we leave the formal details of this implication for future studies.

To conclude, the main conceptual advance of this method is that, once combined with exact structure-preserving reductions, differential moment hierarchies with implicit time-dependence become a practical certification framework for constrained finite-horizon quantum control, yielding computable upper bounds on maximal reachable fidelities and lower bounds on the minimum control times needed to achieve them. 

\section{Results}
\label{sec:results}

We now apply the certification framework to three representative quantum control tasks: entangled-state preparation, one-qubit gate synthesis under restricted control geometries, and end-to-end excitation transfer in a boundary-controlled $XX$ chain. 

For all benchmarks considered, the central technical step is an \emph{exact} reduction of the original dynamics before relaxation. The SDP values obtained from \eqref{eq:SDP_reach} and \eqref{eq:sdp_time} are therefore rigorous bounds for the original control problems after rewriting them in the corresponding exact reduced variables, where any loss of tightness comes from the SDP relaxation itself, and not from the problem reformulation.
By contrast, the lower bounds on reachable fidelities are obtained from explicit feasible pulses found numerically by a GRAPE-inspired method~\cite{Khaneja2005GRAPE}. Concretely, for each fixed horizon we parametrize the controls as piecewise-linear functions with values on a uniform time grid, and optimize the terminal fidelity over these values. The gradient is approximated by differentiating the discretized dynamics, where implementation details vary with each example. We normalize the obtained gradient and project back onto the admissible pulse set. Since this is a local nonconvex search, we use several random initial pulses and keep the best feasible trajectory found. These pulses provide lower bounds on reachable fidelities and feasible upper bounds on threshold times, while all impossibility statements come from the SDP relaxations. We employ the commercial semidefinite solver MOSEK \cite{mosek} to solve the resulting SDPs for each benchmark. The code to reproduce our results is available on GitHub \footnote{\url{https://github.com/nininaceur/ReachabilityQuantumControl}}.

\subsection{State preparation}
We first consider finite-time preparation of the Bell state $\ket{\Phi^+}=\frac{1}{\sqrt{2}}(\ket{00}+\ket{11})$ from the two-qubit product state $\ket{\psi(0)}=\ket{00}$. We use two continuous controls $u_x(t),u_y(t)\in W^{1,\infty}([0,T],\mathbb{R})$ subject to amplitude and velocity constraints. The fidelity $\mathcal{F}_{\Phi^+}(T)$ to the target state at time $T$ can be expressed as a polynomial in Pauli operators $O=\{X_i,Y_i,Z_i\}_{i=1,2}$, which we use interchangeably with $\{\sigma^x_i,\sigma^y_i,\sigma^z_i\}_{i=1,2}$,
\begin{align*}
\mathcal{F}_{\Phi^+}(T) &:= \braket{\psi(T)|\Phi^{+}}\braket{\Phi^{+}|\psi(T)}\\
    &= \frac{1+\bra{\psi(T)}X_1X_2-Y_1Y_2 + Z_1Z_2\ket{\psi(T)}}{4}.
\end{align*}
More generally, arbitrary $N$-qubit target-state fidelities can be written as polynomial observables because Pauli strings span $\mathrm{Herm}(\mathbb{C}^{2^N})$. Before exploiting structure, and after switching to the Heisenberg picture, the fixed-$T$ reachability problem reads

\begin{align}
\small
    &\sup_{\{u_x,\,u_y,\,v_x,\,v_y,\,O\in \{X_i,Y_i,Z_i\}\}}\mathcal{F}_{\Phi^+}(T)
    \label{eq:2q_bell_reachability}\\
    \text{s.t.}\quad
    &\frac{d}{dt}O(t)
    = 
    i\,[H(t),O(t)],
    \qquad t\in[0,T],
    \nonumber\\
    &H(t)
    =
    H_0+u_x(t)H_x+u_y(t)H_y,
    \nonumber\\
    &H_0=\frac{J}{2}Z_1Z_2,\quad
    H_x=\frac{1}{2}(X_1+X_2),\nonumber \\
    &H_y=\frac{1}{2}(Y_1+Y_2),
    \nonumber\\
    &\frac{du_x(t)}{dt}=v_x(t),\qquad
    \frac{du_y(t)}{dt}=v_y(t),
    \nonumber\\
    &u_x(t)^2+u_y(t)^2\le U_{\max}^2,
    \quad
    v_x(t)^2+v_y(t)^2\le V_{\max}^2,
    \nonumber\\
    &
    u_x(0)=u_y(0) = 0,
    \nonumber \\
    &\ket{\psi} = \ket{00}\nonumber.
\end{align}
Throughout this section, we set $J=1,U_{\max}=1,$ and $V_{\max}=2$.

Now, to phrase this problem as a noncommutative polynomial optimization problem as in \eqref{eq:SDP_reach}, we consider the Pauli operators as formal variables subject to algebraic $SU(2)$ constraints
\begin{equation}
    [\sigma_i^a,\sigma_j^b] = 2i\epsilon_{abc}\delta_{ij}\sigma_i^c, \qquad (\sigma_i^a)^2 = 1,
\end{equation}
and work in the quotient algebra $\mathbb{K}\langle\underline{X},\tau \rangle / \mathcal{I}$ over the ideal $\mathcal{I}$ generated by the constraints. We restrict the Pauli sector to monomials up to degree $2$, as higher-degree monomials collapse under the quotient relations. The relaxation order $d$ hence only accounts for the maximal degree of the commuting moments, and the indexing basis for the Pauli sector remains fixed for all $d$. To impose the initial condition $\ket{\psi} = \ket{00}$ we impose that the $\tau=0$ pseudo-moments  coincide with the true moments of $\ket{00}$ for all monomials $m \in \mathbb{K}\langle O,u_i,\tau \rangle_{2d}/\mathcal{I}$:
\begin{equation}
    L_{y(0)}(m) = \bra{00} m\ket{00}
\end{equation}
Now, the size of the SDP can be further reduced by exploiting additional exact symmetries of the benchmark. First, both the drift $H_0$, the controls $H_X,H_Y$, the initial state $\ket{00}$, and the target projector onto $\ket{\Phi^+}$ are invariant under the qubit-swap operator $\mathrm{SWAP}$
\begin{equation}
\label{eq:swap}
\mathrm{SWAP}=\frac{1}{2}\bigl(I+X_1X_2+Y_1Y_2+Z_1Z_2\bigr).
\end{equation}
 We may therefore restrict the operator variables to the swap-invariant algebra, i.e.\ to operators $O$ satisfying $[\mathrm{SWAP},O]=0$. 

Secondly, since $\ket{00}$ lies in the triplet sector and the dynamics preserve that sector, the following relation holds
\begin{equation}
\mathrm{SWAP}\ket{\psi}=\ket{\psi},
\end{equation}
which, using \eqref{eq:swap} yields another constraint that can be exploited to reduce the SDP size:
\begin{equation}
\bigl(X_1X_2+Y_1Y_2+Z_1Z_2-I\bigr)\ket{\psi}=0.
\end{equation}
Finally, within the triplet-reduced algebra, we exploit an additional exact $\mathbb{Z}_2$ symmetry of the full DNPO problem. Its action on the variables is given by

\begin{equation}
\begin{aligned}
(\tau,u_x,u_y,v_x,v_y,m_P)\mapsto {}&
(\tau,-u_x,-u_y,-v_x,-v_y, \\
&\qquad Z_1Z_2\,m_P\,Z_1Z_2),
\end{aligned}
\end{equation}
where $m_P$ denotes an arbitrary Pauli monomial. Because the problem is invariant under this $\mathbb{Z}_2$ action, every feasible moment functional may be chosen $\mathbb{Z}_2$-invariant. It follows that all odd moments vanish. Now, moment matrix entries are indexed by products $u^*v$, whose parity is the product of the parities of words $u$ and $v$. Hence, entries mixing even and odd basis monomials do not contribute, and the PSD matrix decomposes into two independent blocks corresponding to the even and odd $\mathbb{Z}_2$ sectors. These reductions are exact and are essential for tractability; see Table~\ref{tab:structure_exploitation} for the resulting block-size reductions. 

The same strategy extends naturally to larger multipartite state-preparation tasks. As a representative three-qubit example, we also consider the generation of the GHZ state \cite{GreenbergerHorneZeilinger1989GoingBeyondBell}
\begin{equation}
\ket{\mathrm{GHZ}}=\frac{\ket{000}+\ket{111}}{\sqrt{2}},
\end{equation}
under the globally driven Ising Hamiltonian
\begin{align}
H(t)&=\frac{J}{2}\left(Z_1Z_2+Z_2Z_3 + Z_1Z_3\right)\\&+\frac{u_x(t)}{2}\left(X_1+X_2+X_3\right)\nonumber +\frac{u_y(t)}{2}\left(Y_1+Y_2+Y_3\right),
\end{align}
with the same amplitude and velocity constraints on $u_x,u_y$. The corresponding fidelity objective is again a polynomial in Pauli moments
\begin{align}\small
&\mathcal{F}_{\mathrm{GHZ}}(T)
=
\frac{1}{8}(
1
+\langle Z_1Z_2\rangle_T
+\langle Z_1Z_3\rangle_T 
+\langle Z_2Z_3\rangle_T \\ \small
&+\langle X_1X_2X_3\rangle_T
-\langle X_1Y_2Y_3\rangle_T
-\langle Y_1X_2Y_3\rangle_T
-\langle Y_1Y_2X_3\rangle_T
), \nonumber
\end{align}
so the problem fits the same structure as the two-qubit example. The initial state is again fixed as the product state $\ket{000}$ by imposing its operator moments on the boundary functional $L_{y(0)}$.  To keep the GHZ relaxation tractable, we exploit the full permutation symmetry of the problem under the action of the symmetry group $S_3$, which allows us to restrict to the permutation-invariant operator algebra. In addition, because the dynamics remain in the fully symmetric spin-$3/2$ sector, the operator variables can be further compressed to this irreducible subspace, similarly to the Bell pair benchmark.

Figures 
\ref{fig:bounds_2q_symmetry_d2_d3_extended} and \ref{fig:bounds_3q_ghz_symmetry_d2_d3} show the computed finite-time SDP upper bounds on the terminal fidelities for constrained Bell state preparation and GHZ preparation respectively. The insets show representative feasible control pulses at illustrative horizons $T$ giving rise to the displayed lower bounds.

\begin{figure}[H]
    \centering
    \includegraphics[width=\linewidth]{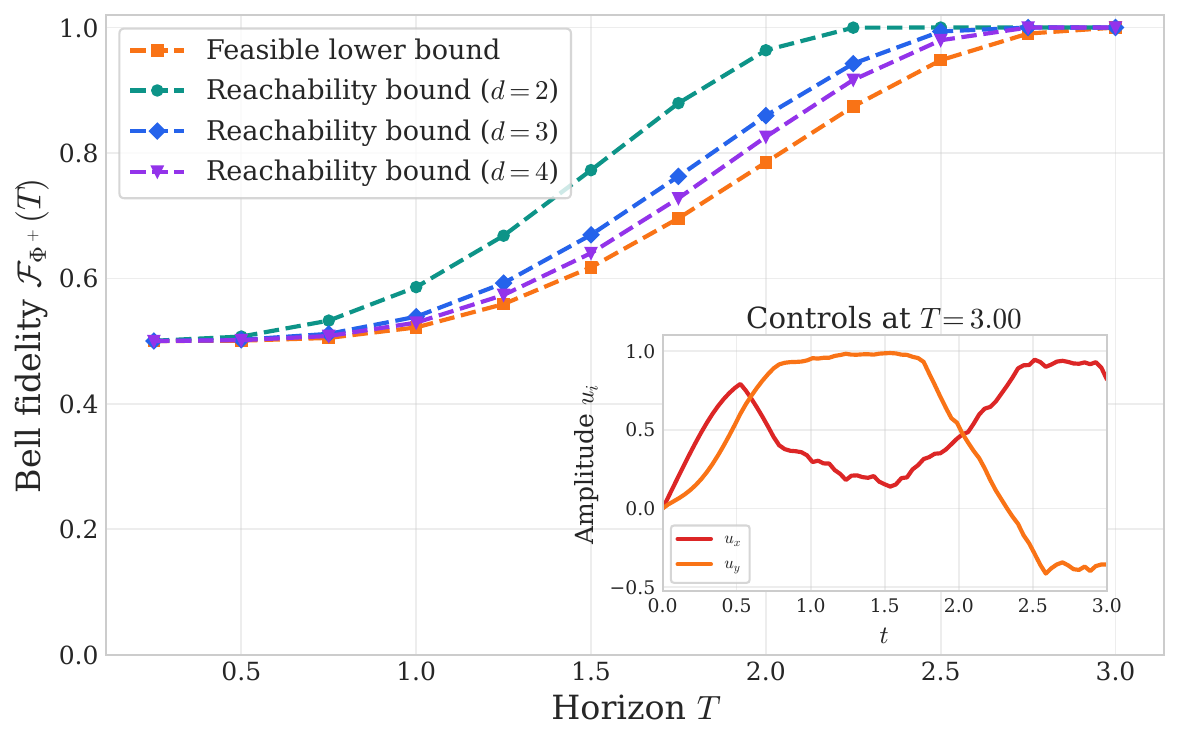}
    \caption{Comparison of the feasible lower bound and the  reachability upper bounds on the target fidelity $\mathcal{F}_{\Phi^+}(T)$ at time $T$ for hierarchy orders $d=2,3,4$. Local optimizer pulses at horizon $T=3.00$ are displayed in the inset.}
    \label{fig:bounds_2q_symmetry_d2_d3_extended}
\end{figure}

\begin{figure}[H]
    \centering
    \includegraphics[width=\linewidth]{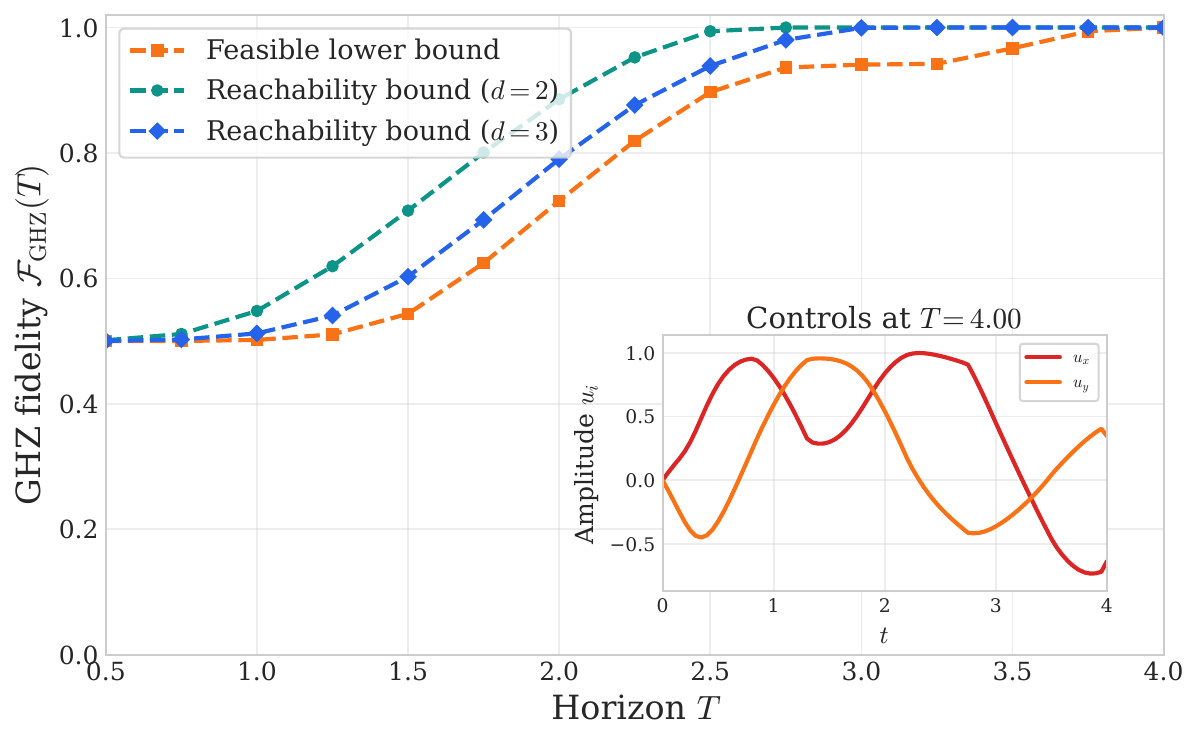}
    \caption{Reachability bounds for the three-qubit GHZ generation, comparing feasible lower bound with the hierarchy upper bound at orders $d=2$ and $d=3$ for different control horizons $T$. The inset shows an illustrative control pulse at $T=4.00$.}\label{fig:bounds_3q_ghz_symmetry_d2_d3}
\end{figure}

For Bell-state preparation, the hierarchy yields nontrivial upper bounds and exclusion regions already at order $d=2$, and the $d=3$ and $d=4$ relaxations substantially tighten the gap to the feasible lower bound, reducing it by more than a factor of two over most intermediate horizons. For the GHZ preparation we likewise find nontrivial bounds at relaxation order $d=2$, and a strong tightening at order $d=3$, where order $d=4$ proved computationally impractical. The rapid growth of the moment matrix with the relaxation order $d$ and the number of dynamical variables is therefore the main practical limitation of the method, and of moment hierarchies for DNPO more generally. 

 Figure \ref{fig:bounds_3q_ghz_symmetry_d2_d3} also shows that the lower bound reaches a plateau before reaching fidelity $\mathcal{F}_{\mathrm{GHZ}}(T)=1$. This is an indicator for a highly complicated optimization landscape that the gradient ascent method struggles to navigate. This interpretation is supported by the good performance of the local optimizer in the other benchmarks, leading us to conjecture that the plateau originates from a complicated optimization landscape, rather than from a limitation of the local method. By employing more sophisticated methods, one could possibly overcome the plateau and improve the lower bound, although here, for consistency and simplicity, we stick to the same gradient method for all applications.

Taken together, the Bell and GHZ benchmarks show that the method can certify genuinely nontrivial finite-time limitations for entangled-state preparation, as well as near-optimality proofs for explicit control pulses, with exact structure reduction playing the key role in turning the hierarchy into a usable computational tool.

\subsection{Gate synthesis}
We now consider the problem of synthesizing a target one-qubit gate $U_\star$ through the time evolution generated by a controlled Hamiltonian $H(t)=H(u(t))$, whose time dependence is governed by continuous control functions 
$u(t)\in W^{1,\infty}([0,T],\mathbb{R})$ subject to experimental amplitude and velocity constraints. This benchmark isolates the role of control geometry: even for a fully controllable one-qubit system, different accessible control directions can lead to different certified finite-time limitations. We ask for the optimal control time to reach a certain threshold $\eta$, where we employ the average gate fidelity $\mathcal{F}_{\mathrm{gate}}(T)$ \cite{Nielsen2002AverageGateFidelity} between the terminal propagator $U(T)$ and the target gate $U_\star$ as the figure of merit
\begin{equation}
\mathcal{F}_{\mathrm{gate}}(T)
=
\frac{|\mathrm{Tr}\left(U_\star^\dagger U(T)\right)|^2+2}{6},
\end{equation}
such that the constraint $\mathcal{F}_{\mathrm{gate}}(T)\geq \eta$ restricts the search space to times $T$ after which fidelity $\eta$ can be attained. In contrast to the state-preparation setting, this function is not naturally the expectation value of a single observable on the system Hilbert space. However, after adequate problem reduction, it matches the problem form \eqref{eq:general_reachability_heisenberg}. The key structural step is to pass from the unitary dynamics on \(SU(2)\) to the induced adjoint action on Pauli operators. Any unitary \(U\in SU(2)\) defines a real orthogonal \(3\times 3\) matrix \(R(U)\) through
\begin{equation}
\label{eq:blochrotation_clean}
\sigma_\alpha(t)=U^\dagger(t)\sigma_\alpha U(t)
=
\sum_{\beta=1}^3 R(U(t))_{\beta\alpha}\,\sigma_\beta.
\end{equation}
The map $U\mapsto R(U)$ is the adjoint representation of $SU(2)$, with image $SO(3)$ and kernel $\pm I$. After quotienting out this physically irrelevant global sign, the gate-synthesis problem can therefore be reformulated exactly in terms of the transfer matrix $R(t):=R(U(t))\in SO(3)$.

In this representation, the gate fidelity becomes
\begin{equation}
\mathcal{F}_{\mathrm{gate}}(T)
=
\frac{\mathrm{Tr}\bigl(R(U_\star)^T R(U(T))\bigr)+3}{6},
\end{equation}
so the fidelity is now linear in the terminal transfer-matrix entries.

To obtain the transfer matrix dynamics, we expand the Hamiltonian as
\begin{equation}
H(t)=\frac{1}{2}\bigl(h_x(t)X+h_y(t)Y+h_z(t)Z\bigr),
\end{equation}
where coefficients $h_\alpha(t)$ contain contributions from both drift and control.
Inserting \eqref{eq:blochrotation_clean} into the Heisenberg equation and matching Pauli coefficients yields
\begin{equation}
\frac{d}{dt}R(t)=R(t)A(t),
\end{equation}
where $A(t)$ is the skew-symmetric matrix
\begin{equation}
A(t)=
\begin{pmatrix}
0 & h_z(t) & -h_y(t) \\
-h_z(t) & 0 & h_x(t) \\
h_y(t) & -h_x(t) & 0
\end{pmatrix}.
\end{equation}

Hence the one-qubit gate-synthesis problem can be formulated exactly in terms of the transfer matrix coordinates \(R_{\alpha\beta}(t)\):
\begin{align}
\small
    &\inf_{\{u_j,\,v_j,\,R_{\alpha\beta}(t)\}}
    \quad T
    \label{eq:1q_gate_optimal_time}\\
    \text{s.t.}\quad
    &\frac{\mathrm{Tr}\bigl(R(U_\star)^T R(U(T))\bigr)+3}{6}\ge \eta,
    \nonumber\\
    &\frac{d}{dt}R(t)=R(t)A(t),
    \qquad t\in[0,T],
    \nonumber\\
    &A(t)=
    \begin{pmatrix}
    0 & h_z(t) & -h_y(t)\\
    -h_z(t) & 0 & h_x(t)\\
    h_y(t) & -h_x(t) & 0
    \end{pmatrix},
    \nonumber\\
    &H(t) = H_0 + \sum_j u_j(t) H_j = \frac{1}{2}\sum_{\alpha\in \{x,y,z\}} h_\alpha \sigma_{\alpha},
    \nonumber\\
    &\frac{du_j(t)}{dt}=v_j(t),
    \qquad j=1,\dots,m,
    \nonumber\\
    &\sum_{j=1}^m u_j(t)^2\le U_{\max}^2,
    \qquad
    \sum_{j=1}^m v_j(t)^2\le V_{\max}^2,
    \nonumber\\
    &u(0)=u_0,
    \qquad
    R(0)=I_3.
    \nonumber
\end{align}
After passing to transfer-matrix coordinates, the one-qubit gate-synthesis problem therefore reduces exactly to a \textit{commutative} polynomial optimization problem with differential constraints. At relaxation order $d$ we  index the moment matrices with commuting monomials in variables $\{\tau, \alpha, u_i(\tau),v_j(\tau),R_{\alpha \beta}(\tau) \}$ up to degree $d$, following Proposition \ref{prop:threshold_time}.

We apply this reduced formulation to two one-qubit control families. The first is a two-axis control model
\begin{equation}
    H_X(t) + H_Y(t) =  \frac{u_x(t)}{2}X + \frac{u_y(t)}{2}Y,
\end{equation}
and the second is a one-axis control model parameterized by an angle \(\theta\),
\begin{equation}
    H_{XY}(\theta,t) = u_{xy}(t)/2\left(\cos(\theta)X + \sin(\theta)Y\right),
\end{equation}
both subject to a fixed drift Hamiltonian
\begin{equation}
    H_0 = \frac{J}{2}Z; \qquad J = 1
\end{equation}
and amplitude- and velocity constraints $U_{\max}=1, V_{\max}=2$.

For each control structure and for target gates \(U_\star\in\{S,X,Y,Z\}\), we bound the optimal time by combining lower bounds on the optimal control time through the hierarchy from Proposition \ref{prop:threshold_time} at  relaxation order $d=3$ and search horizon $\overline T=5.0$, with explicit feasible upper bounds.  Here, 
\begin{equation}
    S=\begin{pmatrix}

1 & 0\\

0 & i

\end{pmatrix}
\end{equation}
denotes the phase gate. The resulting bounds are presented in Table \ref{tab:1q_gate_windows}.

With two independent transverse controls, all four target gates are reached within comparatively short certified windows. In the one-axis family $H_{XY}(\theta)$, the fixed $Z$-drift makes phase gates such as $S$ and $Z$ consistently accessible, with only mild dependence on the transverse control angle. By contrast, the transverse Pauli gates $X$ and $Y$ exhibit substantially broader windows, and in several one-axis cases the fidelity threshold is not reached within the explored horizon $\overline{T}=5.0$. For some entries with $T_{\mathrm{LB}}=5.00$ the SDP was numerically inconclusive, where lower bounds were instead certified via the infeasibility of the objective-free feasibility problem. Moreover, the relative ordering of the $X$- and $Y$-gate windows is inverted with respect to the $\theta$-parametrized control: tuning $\theta$ in the direction of a pure $X$-gate improves the certified window for $Y$, whereas tuning the control towards $Y$  improves the certified window for $X$. These results show that the relevant obstruction is not only whether the system is controllable in principle, but how efficiently a given hardware control geometry can realize a given target on a finite horizon. In particular, already at the one-qubit level, restricting the accessible control axis produces quantitatively different finite-time limitations for different target gates. This illustrates that optimal-time certification captures operational limitations that are directly relevant for gate implementation under hardware constraints.

\begin{table}[t]
\centering
\small
\renewcommand{\arraystretch}{1.1}
\setlength{\tabcolsep}{1pt}
\begin{tabular}{lcccc}
\toprule
\textbf{Control} & $\boldsymbol{S}$ & $\boldsymbol{X}$ & $\boldsymbol{Y}$ & $\boldsymbol{Z}$ \\
\midrule
\textbf{$H_X + H_Y$}
& $[1.27,\,1.30]$
& $[2.92,\,3.17]$
& $[2.92,\,3.17]$
& $[2.40,\,2.48]$ \\

\textbf{$H_{XY}(\pi/2)$}
& $[1.33,\,1.33]$
& $[4.22,\,4.52]$
& $[5.00,>5.00]$
& $[2.90,\,2.90]$ \\

\textbf{$H_{XY}(\pi/3)$}
& $[1.33,\,1.33]$
& $[4.49,\,>5.00]$
& $[5.00,>5.00]$
& $[2.90,\,2.90]$ \\

\textbf{$H_{XY}(\pi/6)$}
& $[1.33,\,1.33]$
& $[5.00,>5.00]$
& $[4.53,\,4.72]$
& $[2.90,\,2.90]$ \\

\textbf{$H_{XY}(0)$}
& $[1.33,\,1.33]$
& $[5.00,>5.00]$
& $[4.22,\,4.52]$
& $[2.90,\,2.90]$ \\
\bottomrule
\end{tabular}

\caption{
Certified gate-synthesis time windows across several one-qubit control families. 
Each entry shows the interval $[T_{\mathrm{LB}},T_{\mathrm{UB}}]$, where $T_{\mathrm{LB}}$ is a lower bound obtained from the cumulative SDP certificate \eqref{eq:sdp_time}, and $T_{\mathrm{UB}}$ is the best feasible upper bound found through explicit pulses.
Entries of the form $[\,\cdot\,, >5.00]$ indicate that the fidelity threshold $F \ge 0.99$ was not reached within the explored horizon $\overline T=5.00$. Similarly, lower bound values $T_{\mathrm{LB}}=5.00$ indicate that the cumulative SDP certified infeasibility until the explored horizon.
}
\label{tab:1q_gate_windows}

\end{table}

\subsection{State transfer}\label{sec:statetransfer}
Finally, we consider finite-time transfer of a single excitation across an \(N\)-qubit \(XX\) chain, from the left boundary site to the right boundary site, under controlled boundary-couplings as modeled by the Hamiltonian
\begin{align}
        H(t)
    &=
    J\sum_{j=2}^{N-2}\frac{X_jX_{j+1}+Y_jY_{j+1}}{2}
    +u_L(t)\frac{X_1X_2+Y_1Y_2}{2}
    \nonumber\\&+u_R(t)\frac{X_{N-1}X_N+Y_{N-1}Y_N}{2},  
\end{align}
whose time dependence is governed by amplitude and velocity constrained control functions \(u_L(t),u_R(t)\in W^{1,\infty}([0,T],\mathbb{R})\) on the outer edges of the chain, as displayed in Figure \ref{fig:chain_schematic}. 

\begin{figure}[h]
    \centering
    \includegraphics[width=\columnwidth,trim=0 6cm 0 10cm,clip]{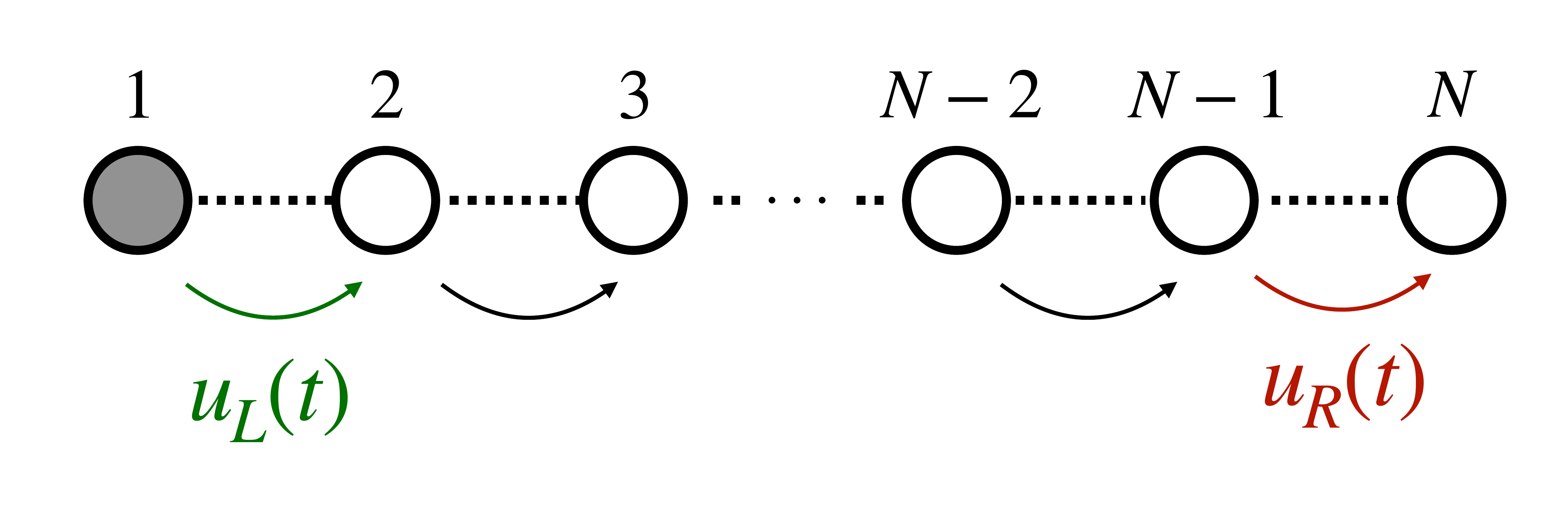}
    \caption{
    Schematic of the reduced single-excitation chain used in the state-transfer benchmark. The two boundary couplings are controlled through the time-dependent amplitudes $u_L(t)$ and $u_R(t)$. The initial excitation is localized at the left end, $\ket{1}$, and the target state is localized to the right end, $\ket{N}$.
    }
    \label{fig:chain_schematic}
\end{figure}
As the dynamics generated by $H(t)$ preserve the total number of excitations in the chain, we work in a localization basis where $\ket{i}$ denotes the state where the $i$-th site is excited:
\begin{equation}
\ket{i}
:=
\ket{
\underbrace{0\cdots 0}_{i-1}
\,1\,
\underbrace{0\cdots 0}_{N-i}
}, \quad i=1,\dots,N.
\end{equation}
This lets us rewrite the Hamiltonian in that basis
\begin{align}
\small
    H(t) =\, & J\sum_{j=2}^{N-2}\ket{j}\bra{j+1}+\ket{j+1}\bra{j} \\
    &+u_L(t)\left(\ket{1}\bra{2} + \ket{2}\bra{1} \right)\nonumber \\
    &\hspace{10pt}+u_R(t)\left(\ket{N-1}\bra{N} + \ket{N}\bra{N-1} 
     \right) , \nonumber
\end{align}
making the conserved excitation number even more apparent.

As a figure of merit for the state transfer we use the transfer fidelity $\mathcal{F}_{\mathrm{tr}}(T)$ to the target state $\ket{N}$ at time $T$,
\begin{equation}
\mathcal{F}_{\mathrm{tr}}(T)=|\braket{N|\psi(T)}|^2.
\end{equation}
The goal is to certify, for given horizons $T$, the largest transfer fidelity achievable by any admissible boundary-control protocol, as well as limits on the optimal control times to reach given target fidelities.

The key exact reduction is the restriction to the single-excitation sector. Because the dynamics preserve excitation number, every reachable state lies in the $N$-dimensional subspace spanned by $\{\ket{i}\}_{i=1}^N$, rather than in the full $2^N$-dimensional Hilbert space. The state-transfer benchmark is therefore no longer formulated over arbitrary operator realizations, but over amplitudes in the reduced basis $\{\ket{i}\}_{i=1}^N$. This exact reduction is what turns the many-body chain into a tractable \textit{commutative} optimization problem under differential constraints. 

Now, to further reduce the complexity of the problem, we exploit the fact that the reduced Hamiltonian only contains nearest-neighbor hopping terms, and therefore couples sites of opposite parity only. We define the diagonal unitary $D$
\begin{equation}
D=\mathrm{diag}(1,i,1,i,\dots).
\end{equation}
Since every nonzero matrix element of $H(t)$ connects an odd site with an even site, conjugation by $D$ multiplies each hopping amplitude by a factor $\pm i$. The transformed Hamiltonian
\begin{equation}
\widetilde H(t):=D^\dagger H(t)\,D
\end{equation}
is therefore purely imaginary. As it remains Hermitian, it can be written as
\begin{equation}
\widetilde H(t)=i\,A(t),
\end{equation}
with $A(t)$ real and skew-symmetric. Writing $\ket{\widetilde\psi(t)}:=D^\dagger \ket{\psi(t)}
$, the Schr\"odinger equation becomes
\begin{equation}
\frac{d}{dt}\ket{\widetilde\psi(t)}=A(t)\ket{\widetilde\psi(t)},
\end{equation}
which is a real linear differential equation. Since $D\ket{1}=\ket{1}$, the full trajectory remains in $\mathbb{R}^N$. Thus the state-transfer problem can be formulated exactly in terms of real amplitudes $\ket{\widetilde{\psi}}\in \mathbb{R}^N$, without altering the original problem. The reduced reachability optimization problem then reads

\begin{align}
\small
    &\sup_{\{u_L,\,u_R,\,v_L,\,v_R,\,\ket{\widetilde\psi(t)}\in \mathbb{R}^N\}}
    \mathcal{F}_{\mathrm{tr}}(T)
    \label{eq:chain_transfer_reachability_real}\\
    \text{s.t.}\quad
    &\frac{d}{dt}\ket{\widetilde\psi(t)}
    =
    A(t)\ket{\widetilde\psi(t)},
    \qquad t\in[0,T],
    \nonumber\\
    &A(t)
    =
    -i\widetilde{H}(t)
    \nonumber\\
    &\frac{du_L(t)}{dt}=v_L(t),
    \quad
    \frac{du_R(t)}{dt}=v_R(t),
    \nonumber\\
    &u_L(t)^2+u_R(t)^2\le U_{\max}^2,
    \quad
    v_L(t)^2+v_R(t)^2\le V_{\max}^2,
    \nonumber\\
    &\ket{\widetilde\psi(0)}=\ket{1},
    \qquad
    u_L(0)=u_R(0)=0,
    \nonumber\\
    &\ket{\widetilde\psi(t)}\in\mathbb{R}^N,
    \qquad
    \mathcal{F}_{\mathrm{tr}}(T)=|\bra{N}{\widetilde\psi(T) \rangle}|^2,
    \nonumber
\end{align}
which suits the form of our framework \eqref{eq:general_reachability_heisenberg}. One can alternatively fix the constraint $\mathcal{F}_{\mathrm{tr}}(T)\geq \eta$ and ask for the optimal control time to reach it, by making use of Proposition \ref{prop:threshold_time}.

We employ illustrative Hamiltonian parameters $J=1, U_{\max}=1,V_{\max}=1.5$. To obtain the order-$d$ relaxation  of this problem we index the moment matrices with words of maximal length $d$ composed of variables $\{\tau,u_L,\,u_R,\,v_L,\,v_R,\,\ket{\widetilde\psi(t)}\}$ and impose the necessary constraints on the level of moments. For the lower time bounds, we also add $\alpha$ to the indexing set.

Figure \ref{fig:bounds_transfer_N4_inset} shows the resulting upper bounds on terminal fidelities for $N=4$ at relaxation orders $d=2,3,4$ against a feasible lower bound computed through explicit control pulses. We find nontrivial bounds on the transfer fidelity and monotonic improvement of the upper bound by increasing the relaxation order for all explored $T$. For both early and late horizons, the gap between the explicit pulse and our upper bounds almost closes, suggesting near-optimality of the explicit control pulses in those regimes. At $T=3.25$ the local solver performs near-optimally, while the optimality gap widens again until $T=4.25$. This behavior indicates the local optimizer terminating in local minima rather than the semidefinite relaxation becoming looser.

\begin{table*}[t]
\centering

\small
\renewcommand{\arraystretch}{1.15}
\setlength{\tabcolsep}{4pt}
\begin{tabular}{lcccl}
\toprule
\textbf{Experiment} & \textbf{Case} & \textbf{Initial block size} & \textbf{Reduced block size } & \textbf{Exact structures exploited} \\
\midrule
\textbf{2Q Bell} 
& $d=2$
& $336$
& $97$
& swap symmetry; triplet invariance; sign parity \\

& $d=3$
& $896$
& $256$
& \\

& $d=4$
& $2016$
& $574$
& \\
\midrule
\textbf{3Q GHZ} 
& $d=2$
& $1344$
& $336$
& permutation symmetry; symmetric spin-$3/2$ \\

& $d=3$
& $3584$
& $896$
&  \\
\midrule
\textbf{1Q gate} 
& $d=3$ & $1344$ & $560$ & $SU(2)\mapsto SO(3)$ reduction, $SO(3)$ quotient\\
\midrule
\textbf{XX chain} 
& $N=4,\ d=2$ & $7168$ & $55$ & single-excitation reduction; real-state gauge\\
& $N=4,\ d=3$ & $21504$ & $220$ &  \\
& $N=4,\ d=4$ & $53760$ & $715$ &  \\
& $N=20,\ d=2$ & $\approx 3.1\cdot 10^{13}$ & $378$ &  \\
\bottomrule
\end{tabular}
\caption{
Effect of structure exploitation on the maximal PSD block size in the different examples.
}
\label{tab:structure_exploitation}
\end{table*}

\begin{figure}[h]
    \centering
    \includegraphics[width=\linewidth]{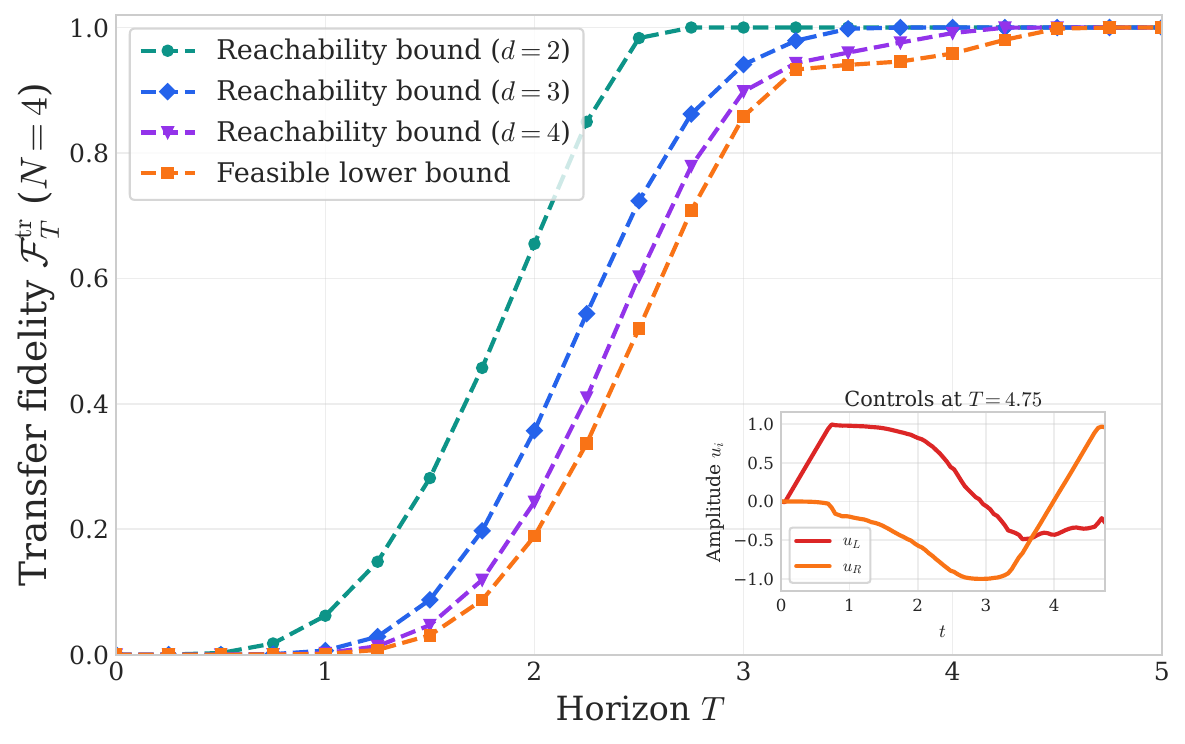}
    \caption{Certified state-transfer fidelity bounds for the spin chain with $N=4$ sites as a function of the horizon $T$. The figure compares the feasible lower bound with hierarchy upper bounds at orders $d=2,3,4$. The inset shows a representative pair of control pulses $(u_L(t),u_R(t))$ for the feasible protocol at $T=4.75$.}
    \label{fig:bounds_transfer_N4_inset}
\end{figure}

\begin{figure}[h]
    \centering
    \includegraphics[width=\linewidth]{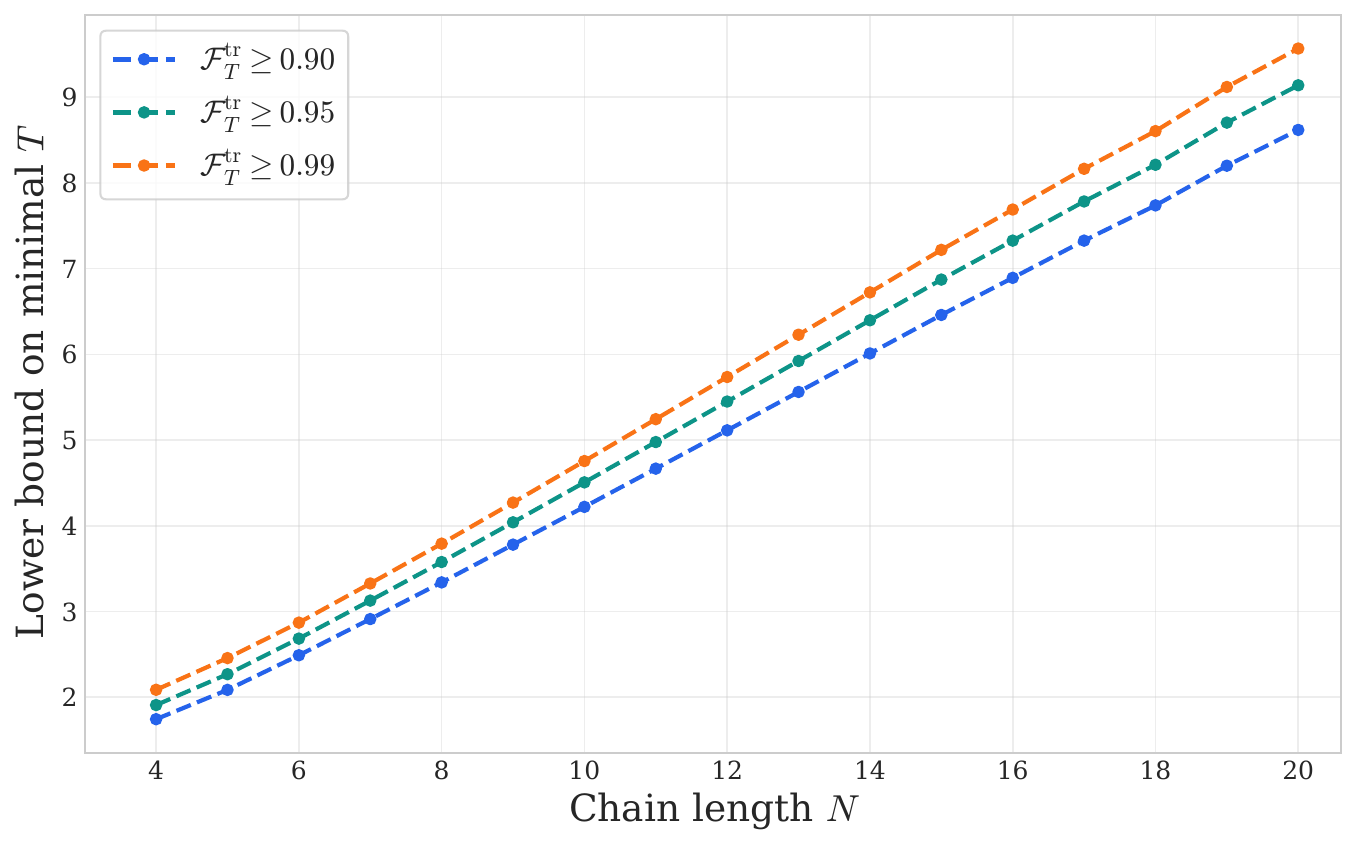}
    \caption{Certified lower bounds on the optimal transfer time $T$ for different target fidelities in the $XX$ chain as a function of the chain length $N=4,\dots,20$, obtained from the reduced cumulative hierarchy \eqref{eq:sdp_time} at relaxation order $d=2$. Each line considers a different fidelity threshold.}
    \label{fig:chain_transfer_T_vs_N}
\end{figure}

Now, by applying the cumulative hierarchy from Proposition \ref{prop:threshold_time}, one can obtain lower bounds on the optimal time $T_\eta^\star$ to reach a desired terminal fidelity $\mathcal{F}_{\mathrm{tr}}=\eta$. To be able to investigate the dependence of the lower bounds on the minimal transfer time $T_\eta^\star$ on $N$ for nontrivial chain lengths $N=4,..,20$, we employ only relaxation order $d=2$ for this experiment. Here, the maximal search horizon $\overline{T}$ is empirically scaled with the chain length, where we observe that the quality of the optimal-time lower bounds $T_\eta^d(\overline{T})$ depends strongly on the choice of $\overline{T}$. This is consistent with the horizon dependence of the cumulative relaxation: increasing \(\overline T\) changes the polynomial constraints through the scaling \(s=\alpha \overline T\). As a result, the SDP may admit additional pseudo-moment trajectories whose projection corresponds to the same physical time \(s\), even though the exact reachable set at \(s\) is independent of \(\overline T\).

As shown in Fig.~\ref{fig:chain_transfer_T_vs_N}, the certified lower bounds on the threshold transfer time increase approximately linearly with the chain length over the explored $N$-range. This trend is consistent with the locality structure of the model; for finite-range quantum spin systems, Lieb--Robinson bounds imply an effective finite propagation velocity \cite{lieb1972finite,nachtergaele2010liebrobinsonboundsquantummanybody}: the influence of an operator initially localized near one end of the chain on an operator localized near the other end is exponentially suppressed outside a light cone whose radius grows linearly in time. Consequently, substantial end-to-end transfer cannot occur in sublinear time as a function of the distance between the endpoints. Thus, a threshold time that grows proportionally to $N$ is the natural behavior to expect. The near-parallelism of the three curves further suggests that, on the sampled range, increasing the target fidelity from $0.90$ to $0.99$ adds a comparatively small extra time on top of this dominant distance-dependent contribution. 

Taken together, these results show that our framework can certify hardware-dependent speed limits for state transfer, including an approximately linear scaling of the threshold time with the chain length $N$ over the explored range.

\section*{Conclusion}
We have developed a certification framework for finite-time constrained quantum control that produces rigorous upper bounds on reachable terminal fidelities and, through an associated SDP relaxation, rigorous lower bounds on the optimal times required to reach them. In contrast to local search methods, these certificates help distinguish genuine dynamical limitations from artifacts of control parametrization or local optimization. From this perspective, the methods developed here not only bound target quantities, but also clarify how control geometry, symmetry, and locality shape the reachable set. This is especially important in constrained finite-time settings, where the failure of a numerical search method cannot by itself be interpreted as evidence of impossibility. At the same time, any remaining gap between upper and lower bounds limits how precisely the true optimum can be identified, since such a gap does not by itself reveal whether it comes from looseness of the SDP upper bound or from the inability of the local solver to find a near-optimal feasible control. More generally, the framework provides a way to assess constrained control architectures not only by constructing pulses that work, but by certifying which fidelities and times are fundamentally excluded.

Moreover, our results show that rigorous finite-time reachability and optimal-time statements in constrained quantum control become practical only when certification methods are combined with exact structure-preserving physical modeling. Across all considered examples, the decisive ingredient in making these bounds computable was the identification of reduced descriptions that preserve the original dynamics while keeping the optimization tractable. In this sense, Table \ref{tab:structure_exploitation} is not merely a numerical summary, but evidence that exploiting exact structure is central to scalable certification in quantum control.

Several promising research directions remain open. An important next direction is to extend the framework to less structured control settings, involving operator-valued controls, where reductions of the kind used here may no longer be available. A second direction is minimizer extraction from these SDP relaxations. In favorable cases, flatness of the moment data may allow direct recovery of optimal controls, while without convergence one could attempt approximate extraction from truncated moments using, for example, occupation-measure support reconstruction \cite{ClaeysDaafouzHenrion2016ModalOccupation} or Christoffel-Darboux semialgebraic recovery \cite{MarxPauwelsWeisserHenrionLasserre2021CD}. The former case would require proving convergence of the controlled DNPO hierarchy, a question that is both of independent interest and nontrivial in view of the technical complexity of the corresponding commutative proofs \cite{henrion2024occupation}.

A further important challenge is robust control; once one moves beyond the nominal setting considered here, uncertainty and model mismatch should be incorporated directly into the certification framework, so that the resulting bounds remain meaningful not only for an idealized model, but also for experimentally relevant imperfect dynamics. Similarly, it would be interesting to investigate the possibility of obtaining reachability bounds on open quantum systems. Finally, it will be important to strengthen the reliability of the numerical output itself by deriving exact certified bounds through rationalization techniques, thereby turning high-precision SDP solutions into rigorous algebraic certificates \cite{naceur2025certifiedboundsoptimizationproblems,cohn2024optimalitysphericalcodesexact,Cafuta2015RationalSO}. 
\vspace{1em}
\section*{Acknowledgements}
The authors would like to thank Victor Magron for insightful discussions. The authors acknowledge the use of GPT 5.5 for assistance with brainstorming ideas and code development. The final content, analysis and conclusions remain the sole responsibility of the authors. This work was granted access to the HPC resources of CALMIP supercomputing center under the allocation 2016-23035. 
 This work has been supported by European Union’s HORIZON–MSCA-2023-DN-JD programme under the Horizon Europe (HORIZON) Marie Skłodowska-Curie Actions, grant agreement 101120296 (TENORS), the European Union Quantera project Veriqtas, the European Union QuantERA project COMPUTE No 101017733, the Government of Spain (Severo Ochoa CEX2019-000910-S, FUNQIP and QEC4QEA PCI2025-163167), the European Union (PASQuanS2.1, 101113690 and QEC4QEA, 101194322), Fundació Cellex, Fundació Mir-Puig, and Generalitat de Catalunya (CERCA program).

\bibliography{references.bib}

\end{document}